%% file: paper.tex
\documentclass[acmsmall, nonacm]{acmart}

\acmJournal{PACMPL}
\startPage{1}

\setcopyright{none}

\usepackage{booktabs}   
\usepackage{subcaption} 
\usepackage{anyfontsize}                        
\usepackage{nameref}
\usepackage{hyperref}
\usepackage{iris}
\usepackage{aneris}
\usepackage{heaplang}
\usepackage{includes}
\usepackage{todonotes}
\usepackage[inline]{enumitem}
\usepackage{amsmath}
\usepackage{nameref}
\usepackage{wrapfig}
\usepackage{array}
\usepackage{mathrsfs}
\usepackage{xspace}
\usepackage{pict2e}
\usepackage{tikz}
\usepackage{cleveref}
\usepackage{fontawesome}
\usepackage{graphicx}
\usetikzlibrary{positioning, calc, matrix, arrows.meta, arrows}
\tikzset{
mymat/.style={
  matrix of math nodes,
  text height=1.67ex,
  text depth=0.75ex,
  text width=2.30ex,
  align=center,
  column sep=-\pgflinewidth
  },
mymats/.style={
  mymat,
  nodes={draw,fill=#1}
  }  
}
\usepackage{caption}

\definecolor{darkblue}{rgb}{0.0, 0.0, 0.55}
\definecolor{darkgreen}{rgb}{0.0, 0.55, 0}
\definecolor{darkred}{rgb}{0.55, 0, 0}

\usepackage[belowskip=-10pt,aboveskip=4pt]{caption}
\include{macros}

\include{transaction_macros}

\allowdisplaybreaks

\begin{document}
\title{Verifying Isolation Levels of Database Implementations for Free Using Separation Logic}

\author{Anders Alnor Mathiasen}
\orcid{0009-0005-6587-5590}                
\affiliation{
  \institution{Aarhus University}            
  \country{Denmark}                          
}
\email{alnor@cs.au.dk}                  

\author{Amin Timany}
\orcid{0000-0002-2237-851X}             
\affiliation{            
  \institution{Aarhus University}
  \country{Denmark}                    
}
\email{timany@cs.au.dk}          

\author{Lars Birkedal}
\orcid{0000-0003-1320-0098}             
\affiliation{
  \institution{Aarhus University}            
  \country{Denmark}                    
}
\email{birkedal@cs.au.dk}          

\begin{abstract}
  Modern databases are highly concurrent and provide transactions as a mean of grouping several database operations 
  into atomically applied units. Database vendors and software engineers use isolation levels to describe the consistency 
  guarantees of transactions. The popular isolation levels give weak guarantees, 
  with intricate semantics, to optimize performance of applications.
  The problem of assuring that database implementations actually implement the isolation level guarantees 
  that application developers build their systems upon has received a great deal of attention from the 
  testing community. But until now, there exists no method for formally verifying that a database implementation 
  actually implements the isolation level that database vendors says it provides.
  In this paper, we present a method for verifying that a database implements an isolation level: 
  we derive isolation levels directly, as formalized in transactional consistency models by the database community,
  from the structure of separation logic specifications.
  By doing so, we consider all program executions that a database and arbitrary clients of the database could produce.
  The result is a so-called free theorem meaning 
  that any database implementation, whose operations are verified against a specific set of separation logic specifications, 
  actually implements its isolation level.
  As all proofs in this paper are mechanized in the Rocq proof assistant and build upon a detailed
  semantic model of program execution, we believe this contribution raises the bar for the 
  achievable robustness of databases.
\end{abstract}

\ccsdesc[500]{Theory of computation~Logic and verification}
\ccsdesc[500]{Theory of computation~Separation logic}
\ccsdesc[500]{Theory of computation~Hoare logic}
\ccsdesc[500]{Theory of computation~Higher order logic}
\ccsdesc[500]{Theory of computation~Database theory}
\ccsdesc[300]{Theory of computation~Distributed algorithms}
\ccsdesc[500]{Theory of computation~Abstraction}

\keywords{Distributed systems, transactions, transactional memory, isolation levels, 
databases, concurrency, separation logic, higher-order logic, 
traces, formal verification, Iris, Rocq} 

\maketitle

\makeatletter
\everydisplay\expandafter{\the\everydisplay \small}
\makeatother


\input{sections/introduction}
\input{sections/overview}
\input{sections/model}
\input{sections/spec}
\input{sections/method}
\input{sections/structure}
\input{sections/proof}
\input{sections/related_work}
\input{sections/conclusion}
\bibliography{paper}
\newpage
\appendix
\input{sections/example}

%

\end{document}

%% file: macros.tex
\makeatletter
\newcommand{\hourglass}{%
  \mathbin{
  \mspace{1mu}
  \mathpalette\vinc@hourglass\relax}
  \mspace{1mu}
}
\newcommand{\vinc@hourglass}[2]{%
  \begingroup
  \settowidth{\unitlength}{$\m@th#1\mspace{6mu}$}
  \begin{picture}(1,1.732)
  \vinc@linethickness{#1}
  \roundcap\roundjoin
  \Line(0,0)(1,1.732)(0,1.732)(1,0)(0,0)
  \end{picture}%
  \endgroup
}
\newcommand{\vinc@linethickness}[1]{%
  \linethickness{%
      \ifx#1\displaystyle 0.8\fontdimen8\textfont3\else
      \ifx#1\textstyle 0.8\fontdimen8\textfont3\else
      \ifx#1\scriptstyle0.8\fontdimen8\scriptfont3\else
      1\fontdimen8\scriptscriptfont3\fi\fi\fi
  }%
}
\makeatother


%% file: transaction_macros.tex
\newcommand{\TC}{\mathbf{C}}

\newcommand{\TR}{\mathbf{R}}
\newcommand{\TW}{\mathbf{W}}

\newcommand{\TTransaction}{\mathit{T}}
\newcommand{\TTransactionSet}[1]{\mathsf{O}_{#1}}

\newcommand{\TTransactionOrder}[1]{<_{#1}}

\newcommand{\TTransactionSetOf}{\mathcal{T}}

\newcommand{\TLevel}{\mathit{I}}

\newcommand{\TState}{s}
\newcommand{\TTransition}[3]{#1 \xlongrightarrow[]{#2} #3}
\newcommand{\TExecution}{\mathit{E}}

\newcommand{\TReadStates}[5]{\mathsf{ReadState}\big(#1, #2, #3, #4, #5\big)}

\newcommand{\TPreread}[2]{\mathsf{Preread}(#1, #2)}
\newcommand{\TComplete}[4]{\mathsf{Complete}(#1, #2, #3, #4)}

\newcommand{\TNoConf}[3]{\mathsf{NoConf}(#1, #2, #3)}
\newcommand{\THasWrite}[2]{\mathsf{HasWrite}(#1, #2)}
\newcommand{\TNoPriorWrite}[2]{\mathsf{NoPriorWrite}(#1, #2)}
\newcommand{\TComTest}[3]{\mathsf{CT}_{#1}(#2, #3)}
\newcommand{\TComTestRU}[2]{\mathsf{CT}_{\mathit{RU}}(#1, #2)}
\newcommand{\TComTestRC}[2]{\mathsf{CT}_{\mathit{RC}}(#1, #2)}

\newcommand{\TComTestSI}[2]{\mathsf{CT}_{\mathit{SI}}(#1, #2)}

\newcommand{\TTags}{\mathsf{tags}}
\newcommand{\TTag}{\mathsf{tag}}
\newcommand{\Tfresh}{\langkw{fresh}}
\newcommand{\Temit}{\langkw{emit}}

\newcommand{\SI}{\scaleto{\mathsf{SI}}{3.5pt}.}
\newcommand{\Ttrace}[1]{\mathsf{Trace}(#1)}

\newcommand{\TtraceInstance}{t}

\newcommand{\TtraceInv}[1]{\mathsf{TraceInv}(#1)}
\newcommand{\ComTrans}[1]{\mathsf{comTrans}(#1)}

\newcommand{\ValSeq}[1]{\mathsf{ValidSeq}(#1)}
\newcommand{\ValSeqSym}{\mathsf{ValidSeq}}
\newcommand{\ValTrans}[1]{\mathsf{ValidTrans}(#1)}
\newcommand{\Trace}{t}
\newcommand{\LTrace}{\overline{t}}
\newcommand{\ValTrace}[2]{\mathsf{ValidTrace}(#1, #2)}
\newcommand{\ExecOf}[2]{\mathsf{CoherenceExec}(#1, #2)}
\newcommand{\CohTrace}[2]{\mathsf{CoherenceTrace}(#1, #2)}
\newcommand{\CohTrans}[2]{\mathsf{CoherenceTrans}(#1, #2)}
\newcommand{\PreFunc}{\mathsf{PreEvent}}
\newcommand{\PostFunc}{\mathsf{PostEvent}}
\newcommand{\PrefixReflect}{\mathsf{Prefix}}

\newcommand{\SIem}{\langkw{emit}}
\newcommand{\SIemC}[1]{\SIem ~ #1}
\newcommand{\SIfr}{\langkw{fresh}}
\newcommand{\SIfrC}[1]{\SIfr ~ #1}
\newcommand{\SIrd}{\textlang{read}}
\newcommand{\SIrdC}[2]{\SIrd ~ #1 ~ #2}
\newcommand{\SIwr}{\textlang{write}}
\newcommand{\SIwrC}[3]{\SIwr ~ #1 ~ #2 ~ #3}

\newcommand{\SIstart}{\textlang{start}}
\newcommand{\SIstartC}[1]{\SIstart ~ #1}
\newcommand{\SIcommit}{\textlang{commit}}
\newcommand{\SIcommitC}[1]{\SIcommit ~ #1}

\newcommand{\SIassert}[1]{\mathsf{assert}(#1)}
\newcommand{\SIletin}[2]{\mathbf{let} ~ #1 \mathrel{=} #2 ~ \mathbf{in}}

\newcommand{\SIconnType}{\mathsf{Conn}}

\newcommand{\SIboolType}{\mathsf{Bool}}
\newcommand{\SIkeyType}{\mathsf{Key}}

\definecolor{commentcolor}{rgb}{0.0, 0.6, 0.18}

\newcommand{\mapstoMemSymb}{\mapsto}
\newcommand{\mapstoMem}[2]{#1 \mapstoMemSymb #2}
\newcommand{\mapstoCacheSymb}[1]{\mapsto_{#1}}
\newcommand{\mapstoCache}[3]{#1 \mapstoCacheSymb{#2} #3}
\newcommand{\KeyUpdStatusSymb}{\mathsf{KeyUpdStatus}}
\newcommand{\KeyUpdStatus}[3]{\KeyUpdStatusSymb(#1, #2, #3)}
\newcommand{\ConnectionStateSymb}{\mathsf{ConnectionState}}

\newcommand{\CanStartSymb}{\mathsf{CanStart}}
\newcommand{\CanStart}[1]{\ConnectionStateSymb(#1, \CanStartSymb)}
\newcommand{\ActiveSymb}{\mathsf{Active}}
\newcommand{\Active}[2]{\ConnectionStateSymb(#1, \ActiveSymb(#2))}

\newcommand{\ipropType}{\mathsf{iProp}}
\newcommand{\statusType}{\mathsf{Status}}
\newcommand{\histType}{\mathsf{Hist}}
\newcommand{\valType}{\mathsf{Val}}
\newcommand{\PerLin}{\mathsf{PermissionLin}}
\newcommand{\PerPost}{\mathsf{PermissionPost}}
\newcommand{\Spec}{\mathit{spec}}

\newcommand{\HistValSymb}{{last}}
\newcommand{\HistVal}[1]{\HistValSymb(#1)}
\newcommand{\HistEmpty}{[~]}
\newcommand{\HistNonEmpty}[1]{[#1]}

\newcommand{\CanCommitPredicateSymb}{\mathsf{CanCommit}}
\newcommand{\CanCommitPredicate}[3]{\CanCommitPredicateSymb(#1, #2, #3)}
\newcommand{\CommitHistSymb}{\mathit{update\_hist}}
\newcommand{\CommitHist}[2]{\CommitHistSymb(#1, #2)}
\newcommand{\MapLookup}[2]{#1[#2]}
\newcommand{\ListLookup}[2]{#1[#2]}

\newcommand{\Wrap}[1]{\mathbf{Wrap}(#1)}
\newcommand{\Wrapres}[1]{\mathsf{Wrap}(#1)}

\newcommand{\Cst}{c}
\newcommand{\Key}{k}
\newcommand{\Keys}{\mathsf{Keys}}
\newcommand{\Hist}{h}

\newcommand{\map}{m}

\newcommand{\Snapshot}{\mathit{ms}}
\newcommand{\Cache}{\mathit{mc}}

\newcommand{\ValueOption}{ov}

\newcommand{\Boolean}{b}

\newcommand{\srvsa}{srv}

\newcommand{\FinishToken}{\mathsf{Tok}}

%% file: sections/introduction.tex
\section{Introduction}
\label{sec:1}

Transactions, the idea of viewing a number of operations as a single unit, are a
cornerstone of database theory and application --- most, if
not all, modern databases support transactions (\eg Oracle, MySQL, Microsoft
SQL Server, PostgreSQL and MongoDB). 
However, when many transactions run concurrently, the consistency guarantees provided by databases 
become subtle and counterintuitive.
Isolation levels describe how isolated concurrent transactions are from each other, and 
it is the canonical way of presenting transactional guarantees. 
Production databases use weak isolation levels,
such as \emph{read committed} and \emph{snapshot isolation},
as default, \eg all the aforementioned databases use 
weak isolation levels. In contrast to strong isolation, which is often referred to as serializability, 
weak isolation levels do not correspond to executing transactions one at time. 
It is important that software engineers who build applications, while assuming these 
isolation level guarantees, 
can be assured that databases actually implement the isolation level vendors claim. 
The database community has a long tradition of formalizing 
and precisely defining weak isolation level guarantees of abstract models of databases
in so-called \emph{transactional consistency models}.
Existing formal correctness proofs of concrete implementations of databases, on the other hand, 
only argue informally that their specifications capture the expected isolation level.

\paragraph{Consistency Models} Much work has been conducted to create transactional consistency models to formally specify transaction isolation guarantees. These models often fall into one of
three categories \citep{Xiong2020}: operational semantics \citep{Xiong2020, transactions-gotsman,
  Crooks17, Kaki17}, abstract executions
\citep{linearizability+strict-serializability, CAP, Burckhardt12, Ketsman23}, or
dependency graphs \citep{Adya99, Adya00}. Historically, the SQL standard of 1992
introduced four isolation levels (informally, an isolation level is a
consistency guarantee between concurrently executing transactions) that SQL
compliant systems must make available to its users: read uncommitted, read
committed, repeatable read, and serializability \citep{orig-serializability}. The definitions of the ANSI SQL
1992 standard were given using so-called ``phenomena'', \ie examples showcasing behavior that should be prohibited when executing concurrent transactions. Critiques that the definitions were
imprecise followed \citep{orig-SI}, leading to the proposal of better
phenomena and another popular isolation level called \emph{snapshot isolation}.
\citet{Adya99} and \citet{Adya00} gave the first widely accepted formal definitions,
using dependency graphs, for a plethora of isolation levels including the ANSI SQL isolation levels and snapshot isolation. 
Though extremely useful, dependency graph models are neither amenable to formal proofs that a concrete implementation of a database is correct with respect to such models, nor are they suitable for formal reasoning about correctness of clients of a database.%
\footnote{Here, by formal verification we mean proving correctness with respect to a detailed semantic model of program execution, \eg using a program logic \`a la Hoare proved sound wrt. such a semantic model.}

\paragraph{Towards Connecting Consistency Models with Executable Database Implementations} 
\citet{Crooks17} present a state-based model for transactional consistency and
prove that their specifications of different isolation levels are equivalent to
those of the recognized dependency graph model. The advantage of the state-based
model, compared to the dependency graph model, is that it is amenable to formal
reasoning about implementations. The recent work of \citet{Soethout2021} makes
the first steps in this direction by checking a \TLA{} specification 
of the state-based model against a database implementation 
(written in the non-executable PlusCal language inside \TLA{}).
\TLA{} is a specification language, used to model local and distributed systems, 
with a model checker to check temporal logic assertions.
While this is a step towards closing the gap between transactional consistency models 
and executable code, we are not fully there yet for a number of reasons:
\begin{itemize}
  \item There is a significant gap between \TLA{} implementations and concrete
  implementations, leaving room for implementation bugs even if the \TLA{} database implementation written in the non-executable PlusCal language itself is correct \citep{paxosLive}.
  \item Correctness in \citet{Soethout2021} means model checking a given transactional workload  against the PlusCal implementation, but checking all possible interleavings of significantly sized workloads is infeasible due to the runtime complexity.
  \item A \TLA{} specification of a system is not suitable for modular reasoning about clients, which may use other libraries in addition to the database API.
\end{itemize}

When it comes to testing databases for correctly implementing isolation levels, 
a splurge of papers have been published recently 
\cite{10.14778/3583140.3583145, 10.1145/3689742, 10.1145/3360591, 10.14778/3430915.3430918, 10.1145/3552326.3567492, 
  moldrup2025awdit, gu2024isovista, pick2025checking}.
The focus of these papers is on \emph{black box testing}, meaning to check whether 
an isolation level is violated in an observed
execution of a transactional workload from a database
(notice the difference between checking one execution of a transactional workload and enumerating all interleavings 
as in \citet{Soethout2021}). For instance,
\citet{pick2025checking} check transactions for isolation levels violations against the model of \citet{Crooks17}, and 
\citet{10.14778/3430915.3430918} uses the model of \citet{Adya00}.
While the application of black-box testing techniques on industrial database systems is impressive, they can only show the presence of bugs,
not the absence of bugs --- to this end, we need to use approaches with stronger guarantees, such as separation logic.

\paragraph{Recent Work in Separation Logic} Separation logic \citep{1029817} has
been very successful in reasoning about the absence of bugs in concrete systems' implementations. In
fact, the state-of-the-art separation logic, the Iris framework \citep{iris,
  iris2, iris3, irisjournal}, has recently been used to give Hoare-style
specifications and correctness proofs for database implementations of the strong
isolation level serializability \citep{vMVCC} as well as the weak isolation
levels read uncommitted, read committed, and snapshot isolation
\citep{SI-placeholder}. 
These papers build upon the Grove \citep{grove}
and Aneris \cite{DBLP:conf/esop/Krogh-Jespersen20}
program logics, respectively, which are both 
instantiations of Iris with a network semantics. 
Grove works with a subset of the Go programming language, and Aneris uses a subset of OCaml.
The specifications in these papers reason about 
real program executions and are highly modular.
In fact, the specifications are parametric in the choice of resources; 
the implementation details of the separation logic resources used in the specifications are 
opaque to the clients using them. 
However, these formal proofs fall short because they
make no connection to any of the transactional consistency models.

\paragraph{Contributions}
In this paper, \textbf{we contribute a method for verifying
executable database implementations against their isolation levels}.
We do so by formally deriving isolation levels,
as formalized in the dependency graph and the state-based models,
 directly
from the structure of the separation logic specifications from \citet{SI-placeholder}.
This is done by adapting the method of creating invariants on execution traces from \citet{free-theorems},
where execution traces are derived from ghost-code instrumentation that can be removed at runtime:
we formulate the state-based consistency model as a predicate on 
execution traces that must invariantly hold.
Using the embedding of the state-based model as an invariant on execution traces, 
\textbf{we establish the connection 
with no additional assumptions about the verified database implementations,
and for any possible usage by clients of the database}. 
We remark that as correctness is proven invariantly, for any client usage of the database,
the correctness result holds for any workload of transactions,
and thus it is a much stronger property than that provided by existing model checking approaches, 
which only work given a specific transaction workload.
The connection is a so-called \emph{free theorem}, as we prove the theorem only by assuming 
that the database implementations satisfy the specifications of \citet{SI-placeholder}.
Thus, our result only has to be proven once for each isolation level
before all database implementations gains the result, this is no matter how complex 
the database implementation may be, \eg it can be distributed over many nodes 
and use advanced locking mechanisms, as long as it satisfies the separation logic specifications.
We provide the proofs for the isolation levels read uncommitted, 
read committed and snapshot isolation in this paper.

Concretely, in this paper we make the following technical contributions:
\begin{itemize}
  \item A formulation and formalization of a realistic version of the state-based transactional consistency 
  model \citep{Crooks17}. Lifting the existing formulation of the model to a more realistic version, 
  useable for reasoning about real transaction workloads, forces us to remove strong simplifying assumptions. 
  \item We embed the state-based transactional consistency 
  model into the Iris separation logic via invariants on traces of executions.
  This is done through multiple transformation phases using insights about linearization points of database operations.
  \item We prove that the separation logic specifications of the weak isolation levels 
  read uncommitted, read committed, and snapshot isolation all adhere to their 
  state-based model specifications by utilizing the specifications' parametricity in separation logic resources. 
  Thus, we achieve modular per-operation specifications on the level of program executions 
  that are verifiably correct with respect to the state-based consistency model and, by the equivalence in \citet{Crooks17}, 
  the dependency graph model, with no observable changes in the specifications or runtime overhead.
  In essence, we derive isolation levels directly from the structure of separation logic specifications.
  \item All our results are mechanized on top of the Rocq proof assistant using the Iris \citep{irisjournal} and Aneris \citep{DBLP:conf/esop/Krogh-Jespersen20} frameworks.
\end{itemize}
Below, we visualize exactly how the contribution of this paper fits into the larger picture; 
the work of \citet{SI-placeholder} invented separation logic specifications for transactional databases 
with weak isolation levels, and showed how to verify that an executable database implementation conforms to
these specifications. This paper shows that any database which conforms to the separation logic 
specifications also conforms to the state-based consistency model of \citet{Crooks17}
without any additional proof obligations.
As the state-based consistency model is proven equivalent to the dependency graph-based
model in \citet{Crooks17}, 
we get that any database implementation verified in separation logic 
is correct as defined in the seminal papers of the dependency graph-based model
\citet{Adya99} and \citet{Adya00}.
\begin{center}
\begin{tikzpicture}[
  box/.style={draw, scale=0.9, rectangle, minimum width=1.4cm, minimum height=1cm, align=center},
  arr/.style={-Stealth, thick, scale=0.9},
  textnode/.style={font=\small, align=center, scale=0.9}
]
\node[box] (db) {\textbf{Database} \\ \textbf{impl.}};
\node[box, right=2.17cm of db] (sep) {\textbf{Sep. log.}\\\textbf{spec.}};
\node[box, right=2.17cm of sep] (crooks) {\textbf{Crooks et al.}\\\textbf{state-based}\\\textbf{model}};
\node[box, right=2.17cm of crooks] (adya) {\textbf{Adya et al.}\\\textbf{dep. graph}\\\textbf{model}};
\draw[arr] (db) -- node[above] {\textit{conforms to}} node[below] {\shortstack{Mathiasen \\ et al. 2025}} (sep);
\draw[arr] (sep) -- node[above] {\textit{implies}} node[below] {This paper} (crooks);
\draw[arr] (crooks) -- node[above] {\textit{equivalent to}} node[below] {\shortstack{Crooks \\ et al. 2017}} (adya);
\end{tikzpicture}
\vspace{-2mm}
\end{center}
To create the implication from separation logic to the transactional consistency 
models, we reuse nothing but the specifications from \citet{SI-placeholder}, 
and the definition of the model from \citet{Crooks17} 
(which we must amend, \cf\ the first technical contribution listed above).
To prove the implication, we use the idea from \citet{free-theorems}
of capturing trace properties solely based on separation logic 
specifications. Thus far, the technique has only been used in a concurrent setting for one
illustrative example. What we do is much more challenging, as
we embed the state-based consistency model of \citet{Crooks17} as a trace property, 
which involves grouping several database operations, each of which has its own individual linearization point, into transactions.

\paragraph{Mechanization} 
All of our results have been mechanized in the Rocq proof assistant, \ie all 
the proofs of our theorems have been machine-checked.\footnote{The Coq proof assistant has recently been renamed to the Rocq prover.}
Together with the fact
that all the proofs of the separation logic specifications of 
\citet{SI-placeholder} have been machine-checked, and the fact that the soundness
of the Iris program logic with respect to the semantic model of program
execution has also been machine-checked, this means that \textbf{we obtain complete
formal theorems in Rocq connecting the database implementations to the state-based
transactional consistency models.}
The Rocq formalization accompanies the paper and consist of more than 13 thousand lines of proof code.


%% file: sections/overview.tex
\section{Technical Overview and Prior Work}
\label{sec:overview}
\paragraph{Database Implementations}
The database implementations we consider in this paper, 
are databases written in AnerisLang. 
This is a general purpose programming language 
and a subset of OCaml with fork-based concurrency, mutable references, 
and networking primitives for communicating over UDP network sockets
(\citet{rel-comm} builds a verified reliable communication library on top of UDP sockets in AnerisLang).
We have included a subset of AnerisLang below, for the full details about the language, we 
refer to its documentation \citep{aneris-documentation} which contains a detailed model 
of execution including the memory and network state.
\begin{align*}
  \expr \in \Expr \bnfdef{}
  & \val
    \ALT \var
    \ALT \expr_{1}~\expr_{2}
    \ALT \HLOp_1 \expr
    \ALT \expr_{1} \HLOp_2 \expr_{2}
    \ALT \If \expr_{1} then \expr_{2} \Else \expr_{3} 
    \ALT
    (\expr_{1}, \expr_{2})
    \ALT \Fst~\expr
    \ALT \Snd~\expr \\
  \ALT & 
    \Inl~\expr
    \ALT \Inr~\expr
    \ALT (\Match \expr with \Inl ~ \var_1 => \expr_1 | \Inr ~ \var_2 => \expr_2 end) \\
  \ALT & 
    \Fork{\expr}
    \ALT \refl ~ \expr 
    \ALT \deref\expr  
    \ALT \expr_{1}\gets\expr_{2}
    \ALT \eCAS{\expr_{1}}{\expr_{2}}{\expr_{3}} 
    \ALT \emakeaddress{\expr_{1}}{\expr_{2}}\\
  \ALT & 
   \egetaddress{\expr} 
   \ALT \enewsocket{}
   \ALT \esocketbind{\expr_{1}}{\expr_{2}} 
   \ALT \esendto{\expr_{1}}{\expr_{2}}{\expr_{3}} 
   \ALT \ereceivefrom{\expr}
   \ALT \dots
\end{align*}
We consider database implementations with operations for starting and committing transactions
together with operations for reading and writing to keys inside a transaction, that is
$\SIstartC{\Cst}$, $\SIrdC{\Cst}{\Key}$, $\SIwrC{\Cst}{\Key}{\val}$ and $\SIcommitC{\Cst}$, 
where $\Cst$ is a connection to the database acquired in an initialization phase.
We make no assumptions about database implementations other than that the
specific separation logic specifications from \citet{SI-placeholder} 
have been proven for these operations; we go over the specifications in \Cref{sec:3}. 
Database implementations may be local or distributed and use all sorts of fancy concurrency measures. 

\paragraph{Crooks Model} 
To verify the isolation levels of database implementations, 
we create a direct formal link to Crooks' state-based transactional consistency model \citep{Crooks17}. 
In the state-based model, isolation levels are formalized as predicates on executions \citep{Crooks17}.
An execution in this context is a sequence of state and transaction pairs, a state is mapping from 
keys of the database to the values they point to. For any pair of a state and a transaction $(\TState_i, \TTransaction_i)$
in an execution $\TExecution$, $\TState_i$ is the result of applying the write operations of 
$\TTransaction_i$ to the previous state $\TState_{i-1}$ in $\TExecution$
(we consider transactions as transitions between states).
For any given isolation levels $\TLevel$, we have a predicate $P_\TLevel$ on executions.
If for all the transactions a database produces, 
we can build an execution $\TExecution$ out of those transactions for which $P_\TLevel(\TExecution)$ is true, 
then the database implements isolation level $\TLevel$.
We formalize a more realistic version of the state-based model in \Cref{sec:2}, where we
remove a few simplifying assumptions from the original presentation to be
able to reason about program executions of database implementations.

\paragraph{Extraction of Transactions through Ghost Code Instrumentation}
To establish that a database implementation written in AnerisLang implements an isolation level 
$\TLevel$, we must extract all the transactions it can produce. Extracting transactions is not trivial, 
as we need to extract any combination of transactions a database possibly can produce by interacting 
with clients programs of the database, and we can not make any assumptions about the behavior of clients. 
The technique of \citet{free-theorems}, which so far only has been applied to small synthetic examples, 
utilizes ghost-code instrumentation to capture trace properties
(everything happens statically, and no code is ever run). 
To use this technique, we have taken AnerisLang and extended its program state with a global trace $\TtraceInstance$ of
so-called \emph{events}. There are two types of events, \emph{pre events} and \emph{post events}, 
marking the start and end of a database operation.
To create such events, AnerisLang is augmented with two ghost-code operations: $\Tfresh$ and $\Temit$.
Using these ghost-code operations, we instrument each of the database operations to emit a
pre event and a post event into the global trace $\TtraceInstance$, 
see \eg the instrumentation of the $\SIrd$ operation below:
\begin{align*}
  \Wrap{\SIrdC{\Cst}{\Key}} \eqdef \SIletin{\tau}{\SIfrC{(\Cst, \mathbf{R}, k)}} ~ \SIletin{ov}{\SIrdC{\Cst}{\Key}} ~ \SIemC{(\tau, \Cst, \mathbf{R}, k, ov)}; ov
\end{align*}
Above, $\SIfrC{(\Cst, \mathbf{R}, k)}$ emits the pre event $(\tau, \Cst, \mathbf{R}, k)$ into the global trace where 
$\mathbf{R}$ is a label signifying that it is a pre event of a $\SIrd$ operation, and $\tau$ is a fresh 
identifier that we use to tie a pair of pre- and post events together; 
$\SIemC{(\tau, \Cst, \mathbf{R}, k, ov)}$ emits the post event $(\tau, \Cst, \mathbf{R}, k, ov)$ 
into the global trace without generating a new identifier, this time we use the identifier $\tau$ from the pre event.
The fundamental insight that we utilize in this paper, 
is that a trace of pre events and post events is enough information to verify that a database implements 
an isolation levels. \textbf{This is the core technical contribution of the paper: we can create an 
invariant on a trace of events that lets us reason about transactions in Crooks model}. 
The invariant is shown in \Cref{fig:extraction_method}, which in essence is our extraction method for extracting any potential 
transactions from a database implementation. 
We explain the details of the invariant, and the formulation of our main theorem (\Cref{def:adequacy}),
in \Cref{sec:4}. Briefly, 
level 1 is our trace of pre- and post-events from the ghost code instrumentation, level 2 
are linearization points created as part of the proof for each pair of pre- and post-events, 
level 3 is grouping of the linearization points into transactions that we can reason 
about in the state-based model at the top level. At the top level, we  
build an execution $\TExecution$ and show that $P_\TLevel(\TExecution)$ 
holds for the isolation level $\TLevel$ that the database implements. 
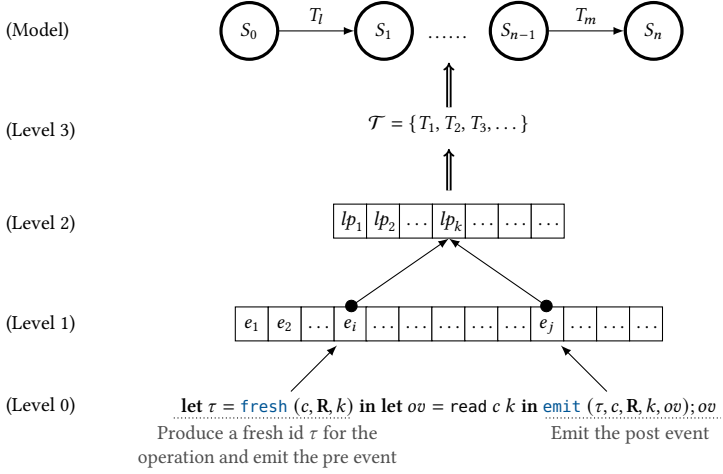
\begin{figure}[!h]
  \centering
  \begin{tikzpicture}[>=latex, scale=0.8, every node/.style={scale=0.75}, node/.style={circle, draw=black, very thick, minimum size=10mm, align=left}]
    \node[node] at (-3.3,3.1) (zero)                  {$S_0$};
    \node[node]               (fst)   [right=of zero] {$S_1$};
    \node[node]               (snd)   [right=of fst] {$S_{n-1}$};
    \node[node]               (rd)   [right=of snd] {$S_{n}$};

    \draw[->] (zero.east) -- (fst.west) node [pos=0.5,above] {$T_l$};;
    \draw[->] (snd.east) -- (rd.west) node [pos=0.5,above] {$T_m$};;
    \node at (0,3.0) {\dots \dots};

    \matrix[mymats=white,anchor=north]
    at (0,0.4) 
    (mat2)
    {
    \mathit{lp}_1 & \mathit{lp}_2 & \dots & \mathit{lp}_k & \dots & \dots & \dots \\
    };
    \matrix[mymats=white,anchor=north]
    at (0,-1.3) 
    (mat1)
    {
    e_1 & e_2 & \dots & e_i &  & \dots & \dots & \dots & \dots & \dots & e_j & \dots & \dots & \dots \\
    };
    \node[below=15pt of mat1]
    (cella-level0) {$\SIletin{\tau}{\SIfrC{(\Cst, \mathbf{R}, k)}} ~ \SIletin{ov}{\SIrdC{\Cst}{\Key}} ~ \SIemC{(\tau, \Cst, \mathbf{R}, k, ov)}; ov$};
    \node[above=20.5pt of mat2]
    (cella-level3) {$\TTransactionSetOf = \{\TTransaction_1, \TTransaction_2, \TTransaction_3, \dots \}$};
    \node at (-6.80,3.1) {(Model)};
    \node at (-6.75,1.45) {(Level 3)};
    \node at (-6.75,-0.1) {(Level 2)};
    \node at (-6.75,-1.75) {(Level 1)};
    \node at (-6.75,-3.1) {(Level 0)};
    \begin{scope}[shorten <= -2pt]
    \draw[*->] (mat1-1-4.north) -- (mat2-1-4.south);
    \draw[*->] (mat1-1-11.north) -- (mat2-1-4.south);
    \draw[-{Implies},double, line width=0.25mm] (0,1.925) -- (0,2.575);
    \draw[-{Implies},double, line width=0.25mm] (0,0.55) -- (0,1.15);
    \draw[->] (2.55,-2.8) -- (1.85,-2.1);
    \draw[->] (-2.55,-2.8) -- (-1.85,-2.1);
    \end{scope}
    \draw[densely dotted] (-4.55, -3.3) -- (-1.60, -3.3);
    \draw[black!70] (-3,-3.3) node[anchor = north] {\begin{minipage}{13em}\centering Produce a fresh id $\tau$ for the operation and emit the pre event\end{minipage}};
    \draw[densely dotted] (1.6, -3.3) -- (4.5, -3.3);
    \draw[black!70] (3,-3.3) node[anchor = north] {\begin{minipage}{10em}\centering Emit the post event\end{minipage}};
  \end{tikzpicture}
  \caption{Extraction method exemplified with the $\mathsf{read}$ operation.} 
  \label{fig:extraction_method}
\end{figure}
\paragraph{Proof Obligations in Separation Logic} 
To show that the invariant holds for all possible usages of a database implementation,
our proof obligation becomes to prove that 
emitting the pre event and emitting the post 
event into the global trace does not brake the invariant for each of the database operations;
where we get to assume that the database operations satisfy their separation logic specifications. 
Thus, in summary, we are proving isolation levels from the structure of separation logic specifications alone.
In \Cref{sec:5}, we present the proof structure in detail and carry out a detailed proof sketch 
for snapshot isolation in \Cref{sec:6}. To explain what we mean by only relying on
the structure of separation logic specifications, 
consider the specifications for a very simple database below: 
\begin{align*}
  \anhoare{\mathsf{idle}}{\SIstartC{\Cst}}{\mathsf{active}}{}{} \hspace{2.6mm}
  \anhoare{\mathsf{active}}{\SIcommitC{\Cst}}{\mathsf{idle}}{}{} \hspace{2.6mm} 
  \anhoare{\mathsf{active}}{\SIwrC{\Cst}{\Key}{\val}}{\mathsf{active}}{}{} \hspace{2.6mm}
  \anhoare{\mathsf{active}}{\SIrdC{\Cst}{\Key}}{\mathsf{active}}{}{}
\end{align*}
The preconditions and postconditions of these specifications only assert whether there
is an active transaction or not.
But, from the structure of these specifications, we can, for instance, deduce 
that the $\SIrd$ and $\SIwr$ operations never can be used before the $\SIstart$ operation, 
and that a transaction can not be committed twice.
The specifications of \citet{SI-placeholder} are, of course, much more involved and they contain enough information 
to let us prove the invariant on the global trace that encapsulates isolation levels in Crooks' state-based model.
\paragraph{Comparison with other Proof-Oriented Approaches for Non-Executable Code} 
Outside the world of separation logic, \citet{raad2018semanticssnapshotisolation} have created 
two high-level non-executable snapshot isolation reference implementations, as local software systems on the 
release/acquire memory model, 
and made a pen and paper proof in their framework that the reference implementations 
are correct regarding the model of snapshot isolation of \citet{analysing_si} (this model is based on dependency graphs 
and is a variant of \citet{Adya99} and \citet{Adya00}).
Similar work has also been done by \citet{rpsi} 
for a variant of snapshot isolation called \emph{parallel} snapshot isolation.
Unlike the approach in this paper, other than being non-mechanized and non-executable, this line of work 
does not address verifying executable clients against the reference implementations. Neither does it come with the modularity 
of separation logic to be usable in conjunction with other software libraries, which makes it unfeasible to use for 
developments of significant size. 
It also remains unclear how one would verify other implementations than the reference implementations 
without having to undergo the same proof effort.
The work of \citet{doherty2013towards} (and comparable approaches like \citet{lesani2012framework} and \citet{schellhorn2025verification}) 
has similar limitations: Transactional guarantees are specified as an I/O automaton, 
and they prove that a transactional memory implementation, also written as an I/O automaton, is a \emph{forward simulation} of the specification. 
Analogously, \citet{Attiya2013} (and the successor work in \citet{attiya2017characterizing}) create \emph{observable refinements} between transactional memory implementations written in a 
non-executable programming language. 
In \cite{lesani2022c4}, a different approach is taken based on \emph{interaction trees} \citep{ITrees}; 
the paper presents a mechanized framework for verifying classical concurrent libraries 
with linearizability \citep{linearizability+strict-serializability} alongside transactional libraries. 
The specifications of the transactional libraries are restricted to serializability formulated as linearizability of transactions 
(informally, serializability is the guarantee that transactions appear to be executed one after another), 
and the approach is applied to \emph{Transactional Mutex Locks} \citep{dalessandro2010transactional}.

In \Cref{fig:comparison-table}, we compare the contributions of this paper with the other proof-oriented approaches for showing the absence of bugs in database or software 
transactional memory (STM) implementations. We compare across six axes: 
(1) \emph{Executeable code} (is the verified database/STM executable?),
(2) \emph{Generalizes over clients} (does the verification prove that all potential transactional workloads created by clients 
of the database/STM break no guarantees?),
(3) \emph{Generalizes over impls.} (do all future verifications have to undergo the same proof effort?),
(4) \emph{Abstract spec.} (is there a formal connection to an abstract transactional consistency model?),
(5) \emph{Modular} (is it a modular proof that can easily be combined with other verified efforts?) and
(6) \emph{Mechanized} (is the proof mechanized in a proof assistant?).
As seen in \Cref{fig:comparison-table}, this paper is the only method that is partially able to generalize over 
database implementations.
Our approach is agnostic to the concrete implementation and only assumes that an implementation is verified against the separation logics specifications 
of \citet{SI-placeholder}, which make no assumptions about implementation details,
or whether the system consists of one or more servers.
Our results establish a connection from the separation logic specifications to the transactional consistency model 
of \citet{Crooks17} stated as a free theorem that does not have to be reproved for other implementations as long as they satisfy the specifications of \citet{SI-placeholder}.
\begin{figure}[h]
  \vspace{-3mm}
  \begin{center}
  \small
  \setlength{\tabcolsep}{3.65pt}
  \begin{tabular}{|c|c|c|c|c|c|c|} 
  \hline
  & \makecell{Executable \\ code} & \makecell{Generalizes \\ over clients} & \makecell{Generalizes \\ over impls.} & \makecell{Abstract\\ spec.} & \makecell{Modular} & \makecell{Mechanized} \\ 
  \hline
  \citet{Attiya2013} & \Circle & \Circle & \Circle  & \Circle & \Circle & \Circle \\ 
  \hline
  \citet{doherty2013towards} & \Circle & \Circle & \Circle  & \Circle & \Circle & \CIRCLE \\ 
  \hline
  \citet{raad2018semanticssnapshotisolation} & \Circle & \CIRCLE & \Circle & \CIRCLE & \Circle & \Circle \\ 
  \hline
  \citet{Soethout2021} & \Circle & \Circle & \Circle  & \CIRCLE & \Circle & \CIRCLE \\ 
  \hline
  \citet{lesani2022c4} & \Circle & \CIRCLE & \Circle  & \Circle & \CIRCLE & \CIRCLE \\ 
  \hline
  \citet{vMVCC} & \CIRCLE & \Circle & \Circle & \Circle & \CIRCLE & \CIRCLE \\ 
  \hline
  \citet{SI-placeholder} & \CIRCLE & \Circle & \Circle & \Circle & \CIRCLE & \CIRCLE \\ 
  \hline
  \textbf{This paper} & \CIRCLE & \CIRCLE & \LEFTcircle   & \CIRCLE & \CIRCLE & \CIRCLE \\ 
  \hline
  \end{tabular}
  \end{center}
  \caption{Comparison of prior work about verification of database and software transactional memory implementations.} 
  \label{fig:comparison-table}
  \vspace{-3mm}
\end{figure}


%% file: sections/model.tex
\section{A Realistic State-based Model}
\label{sec:2}

The original formulation of the state-based model, presented in \citet{Crooks17}, makes a number 
of simplifying assumptions to ease its presentation. While this is a natural thing to do, 
we need to remove some of these assumptions as they are overly restrictive for capturing the behavior 
of clients interacting with actual database implementations. 
Concretely, the assumptions we remove are:
\begin{itemize}
  \item Every value written (as a part of a $\mathsf{write}$ operation) is unique among all other values written.
  \item There is at most one $\mathsf{write}$ operation per transaction for each key.
\end{itemize}
While these simplifying assumptions can seem somewhat minor, removing them requires some definitions 
to be reformulated (we will show how shortly). 
Contrary to the above, there is one assumption we do
not remove: 
\begin{itemize}
  \item All transactions read from their most recent write (if there is one).
\end{itemize}
It is important to understand that we do still prove 
that transactions from a database implementation verified with our approach read from their most recent write,
but we prove it prior to reasoning at the transactional model level. 
We provide more details on this when we explain our method for extracting 
transactions from executable code in \Cref{sec:4}. 
In the rest of this section, we formalize the state-based model, and along the way we explain the 
changes we have made to remove the above-mentioned assumptions. Afterwards, we present a formal definition
of what it means for a database implementation to implement an isolation level in this model.
\paragraph{Transactions} 
Consider transaction 
operations on a set of keys and values.
We work with the transaction operations $\mathsf{read}$, $\mathsf{write}$ and $\mathsf{commit}$.
$\TR(\mathit{id}, k, ov)$
denotes a $\mathsf{read}$ operation reading the optional value $ov$ from key $k$ and having identifier $\mathit{id}$, 
similarly $\TW(\mathit{id}, k, v)$ denotes a $\mathsf{write}$ operation writing value $v$ to key $k$
and having identifier $id$. We assume each operation has a unique identifier.
The $\mathsf{commit}$ operation is
denoted $\TC(\mathit{id}, b)$ with the boolean $b$ indicating whether the transaction commits or aborts.
We define a transaction as a set of operations together with a total order on the operations representing
the program order: 
\begin{definition}[Transaction]\label{def:transaction} 
  A \emph{transaction} $\TTransaction$ is a tuple of 
  $(\TTransactionSet{\TTransaction}, \TTransactionOrder{\TTransaction})$
  where $\TTransactionSet{\TTransaction}$ is a set of operations
  and the program order $\TTransactionOrder{\TTransaction}$ is a total order
  on $\TTransactionSet{\TTransaction}$. 
  $\TTransactionSet{\TTransaction}$ must contain exactly one $\mathsf{commit}$ operation, and 
  it must be the greatest element in $\TTransactionOrder{\TTransaction}$.
\end{definition}
We remark that removing the simplifying assumptions listed in the beginning of this section 
is not enough to allow for all realistic client behavior: 
because transactions are defined as having a \emph{set} of operations, 
we have added the (unique) identifier component to transaction operations, to allow 
for multiple operations having the same keys and values in a single transaction. 
Transaction operations usually do not have an identifier component in other transactional models. 
In \Cref{sec:4}, we show how unique identifiers for transaction operations are created.
\begin{wrapfigure}{R}{0.64\textwidth}
\begin{minipage}{0.64\textwidth}
  \vspace{-3mm}
  \centering
  \begin{tikzpicture}[>=latex, scale=0.8, every node/.style={scale=0.8}, node/.style={circle, draw=black, very thick, minimum size=15mm, align=left},]    
    \draw[->,black,line width=1] (0,0) -- (10,0) node[right] {time};
    \node (0) at (0, 1) {$T_1:$};
    \node (1) [right=of 0] {$\TR(1, x, x_0)$};
    \node (2) [right=of 1] {$\TW(1, y, y_1)$};
    \node (3) [right=of 2] {$\TC(1, \TRUE)$};

    \node (10) at (0, 0.5) {$T_2:$};
    \node (11) [right=of 10] {$\TR(2, y, y_0)$};
    \node (12) [right=of 11] {$\TW(2, x, x_1)$};
    \node (13) [right=of 12] {$\TC(2, \TRUE)$};
  \end{tikzpicture}
  \caption{The \emph{write skew} example.} 
  \label{fig:SI_skew_example}
\end{minipage}
\end{wrapfigure}
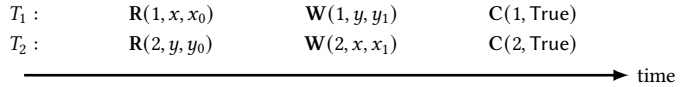 
\paragraph{Isolation levels} 
Having formalized what a transaction is, the next step is to define transactional guarantees, \ie isolation levels. 
We can think of an isolation level as all the executions of a transactional workload we will
allow. 
Throughout this paper, we primarily focus on one particular isolation level, namely snapshot isolation, to explain our ideas. 
Snapshot isolation, together with the two other isolation levels we treat in this paper, read uncommitted and read committed, 
fall into the category of weak isolation levels.
Strong isolation, in the form of serializability, is the isolation level where each transaction gets to execute 
one at a time with no overlapping. Snapshot isolation is more relaxed than serializability and hence weaker. 
The \emph{write skew} example is commonly used to explain how snapshot isolation differs from serializability. 
In the \emph{write skew} example, we observe two concurrent transactions $T_1$ and $T_2$ that both perform 
a $\mathsf{read}$ operation followed by a $\mathsf{write}$ operation before they commit. 
\noindent As seen in \Cref{fig:SI_skew_example}, transaction $T_1$ reads from the key $x$ and writes to the key $y$, whereas it is the opposite for transaction $T_2$.
Both transactions read from the initial snapshot of the database, namely the snapshot where all keys 
have not been written to yet and therefore return $\None$ when read from.
For notational convenience, we often omit the optional part from the result of a $\mathsf{read}$ operation; 
instead, \eg for the key $x$, we write $x_0$ for the initial version 
corresponding to $\None$. 
We denote subsequent versions of a key with an index larger than zero.  
Because neither of the two transactions in the \emph{write skew} example see the effects of the other, this example 
cannot come from an execution where transactions execute one at a time. 
Therefore, it is not allowed under serializability.
But it is allowed under snapshot isolation. 
The principle of snapshot isolation is that \emph{all transactions read from a snapshot of the database, 
and any transaction is allowed to commit when it has no write conflicts with other concurrent transactions} \citep{orig-SI}.
By a write conflict, we mean two transactions that write to the same key. 
Notice that snapshot isolation, and any other isolation level, always have the baseline guarantee,
mentioned in the beginning of this section, that transactions read their own writes.

While an example such as the \emph{write skew} can help us strengthen our intuition and explain the concept of snapshot isolation, 
we need a transactional model to formally decide whether an arbitrary execution of a transactional workload is considered safe. 
In the state-based model, we use a labeled state transition system and, on top of that, we define the notion of \emph{executions}. 

\begin{definition}[States and Executions]\label{def:state}\label{def:transition}\label{def:execution}
  A \emph{state} $\TState$ is a mapping from keys to values.
  We write $\TState[k]$ for the value of key $k$ in state $\TState$.
  We define a \emph{transition relation} for transactions which we write as $\TTransition{\TState}{\TTransaction}{\TState'}$ when the state $\TState'$ results from running (the $\mathsf{write}$ operations of) transaction $\TTransaction$ starting from state $\TState$.
  An \emph{execution} $\TExecution$ 
  is a sequence of state-transaction pairs, $(\TState_1, \TTransaction_1), (\TState_2, \TTransaction_2), \dots, (\TState_n, \TTransaction_n)$ such that $\forall i < n, \TTransition{\TState_i}{\TTransaction}{\TState_{i+1}}$.
  We write $\ListLookup{\TExecution}{i}$ for the pair $(\TState_i, \TTransaction_i)$ in the execution sequence.
\end{definition} 
\vspace{2mm}
\noindent To express that an execution $\TExecution$ is based upon a set $\TTransactionSetOf$ of transactions, 
we write $\ExecOf{\TTransactionSetOf}{\TExecution}$,  
which formally expresses that no transactions occur more than once in an execution $\TExecution$ and 
$\forall \TTransaction, (\exists i, ~ \ListLookup{\TExecution}{i} = (\_,\TTransaction)) \Leftrightarrow \TTransaction \in \TTransactionSetOf$.

For the \emph{write skew} example in \Cref{fig:SI_skew_example}, we show its two possible executions in \Cref{fig:SI_model_example}.
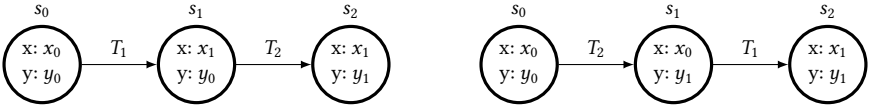
\begin{figure}[h]
  \vspace{-3mm}
  \centering
  \begin{tikzpicture}[>=latex, scale=0.8, every node/.style={scale=0.8}, node/.style={circle, draw=black, very thick, minimum size=12mm, align=left},]
  
    \node[node, label={$s_0$}]      (zero)                      {x: $x_0$ \\ y: $y_0$};
    \node[node, label={$s_1$}]      (fst)       [right=of zero] {x: $x_1$ \\ y: $y_0$};
    \node[node, label={$s_2$}]      (snd)       [right=of fst]  {x: $x_1$ \\ y: $y_1$};
    
    \draw[->] (zero.east) -- (fst.west) node [pos=0.5,above] {$T_1$};;
    \draw[->] (fst.east) -- (snd.west) node [pos=0.5,above] {$T_2$};;
  \end{tikzpicture} 
  \hspace{10mm}
  \begin{tikzpicture}[>=latex, scale=0.8, every node/.style={scale=0.8}, node/.style={circle, draw=black, very thick, minimum size=12mm, align=left},]
  
    \node[node, label={$s_0$}]      (zero)                      {x: $x_0$ \\ y: $y_0$};
    \node[node, label={$s_1$}]      (fst)       [right=of zero] {x: $x_0$ \\ y: $y_1$};
    \node[node, label={$s_2$}]      (snd)       [right=of fst]  {x: $x_1$ \\ y: $y_1$};
    
    \draw[->] (zero.east) -- (fst.west) node [pos=0.5,above] {$T_2$};
    \draw[->] (fst.east) -- (snd.west) node [pos=0.5,above] {$T_1$};
  \end{tikzpicture}
  \caption{The \emph{write skew} example in the state-based model with its two possible executions.} 
  \label{fig:SI_model_example}
  \vspace{1mm}
\end{figure}

\noindent Both executions start from the initial state, here we represent the initial values
as version zero, \eg $x_0$, but they differ in which transaction applies its $\mathsf{write}$ first.
The fact that snapshot isolation allows the \emph{write skew} example means that it contains both executions 
in \Cref{fig:SI_model_example}. We define an \emph{isolation level} to be 
a set of executions.

\paragraph{Commit tests}
Whether an execution is in an isolation level
$\TLevel$ is determined by a \emph{commit test} \citep{Crooks17}, \ie a predicate on executions and transactions, written as $\TComTest{\TLevel}{\TTransaction}{\TExecution}$. 
An execution $\TExecution$, 
for the subset of committed transactions in $\TTransactionSetOf$ (\ie assuming $\ExecOf{\ComTrans{\TTransactionSetOf}}{\TExecution}$ holds 
where $\ComTrans{\TTransactionSetOf}$ is the subset of committed transactions in $\TTransactionSetOf$),
is in an isolation level $\TLevel$
when the test holds for all the commited transactions in 
$\TTransactionSetOf$, \ie
$\forall \TTransaction \in \ComTrans{\TTransactionSetOf}, ~ 
\TComTest{\TLevel}{\TTransaction}{\TExecution}$. 
We define the commit test for the weak isolation levels read uncommitted, read committed and snapshot isolation using a number of predicates.
\Cref{fig:state_based_levels} summarises the definition of the commits tests.

Starting with snapshot isolation, 
we write $\THasWrite{\TTransaction}{k}$ to say that the transaction $\TTransaction$ contains at least one $\mathsf{write}$ operation with key $k$, 
and write $\TNoPriorWrite{\TTransaction}{op}$ to state that in transaction $\TTransaction$ there is no $\mathsf{write}$ operation on the same key as 
the operation $op$ prior to $op$ itself.
Using these definitions we define the predicates $\mathsf{Complete}$ and $\mathsf{NoConf}$ in \Cref{def:predicates_si} which are used 
in the commit test of snapshot isolation. The $\TComplete{\TExecution}{(\TTransactionSet{\TTransaction}, \TTransactionOrder{\TTransaction})}{\TState}{i}$ predicate 
asserts that a transaction $(\TTransactionSet{\TTransaction}, \TTransactionOrder{\TTransaction})$ has a snapshot in the past of the execution $\TExecution$. 
The snapshot is represented by the state $s$, while we use $i$ as the index of $s$ in $\TExecution$.
The $\TNoConf{\TExecution}{\TTransaction}{i}$ predicate asserts that the transaction $\TTransaction$, with snapshot at index $i$ in 
$\TExecution$, has no write conflict with any concurrent transaction. 
Here a concurrent transaction is one that has committed in between the start time 
and the commit time of transaction $\TTransaction$.

\begin{definition}\label{def:predicates_si}
  We define the predicates $\mathsf{Complete}$, and $\mathsf{NoConf}$ as follows:
  {
  \begin{align*}
    \intertext{\normalsize For a transaction $(\TTransactionSet{\TTransaction}, \TTransactionOrder{\TTransaction})$, 
    execution $\TExecution$, state $\TState$, and index $i$}
    \TComplete{\TExecution}{(\TTransactionSet{\TTransaction}, \TTransactionOrder{\TTransaction})}{\TState}{i} \eqdef & ~
    \forall j ~ \mathit{id} ~ k ~ ov, ~
    \ListLookup{\TExecution}{j} = (\_, (\TTransactionSet{\TTransaction}, \TTransactionOrder{\TTransaction})) \Rightarrow\tag{Complete} \label{def:complete}
    \\ & ~\TR(\mathit{id}, k, ov) \in \TTransactionSet{\TTransaction} \Rightarrow
    \TNoPriorWrite{(\TTransactionSet{\TTransaction}, \TTransactionOrder{\TTransaction})}{\TR(\mathit{id}, k, ov)} \Rightarrow \\
    & ~ i < j \land \ListLookup{\TState}{k} = ov.
    \intertext{\normalsize For a transaction $\TTransaction$, execution $\TExecution$, and index $i$}
    \TNoConf{\TExecution}{\TTransaction}{i} \eqdef & ~ 
    \forall j ~ i' ~ \TTransaction' ~ k, ~
    \ListLookup{\TExecution}{j} = (\_, \TTransaction) \Rightarrow \tag{NoConf} \label{def:no_conf}
    \ListLookup{\TExecution}{i'} = (\_, \TTransaction') \Rightarrow \\
    & ~ i < i' < j \Rightarrow 
    \THasWrite{\TTransaction}{k} \Rightarrow
    \neg \THasWrite{\TTransaction'}{k}.
  \end{align*}
  }
\end{definition}
\noindent The commit test of snapshot isolation in \Cref{fig:state_based_levels} 
then states that all transactions 
in a valid execution $\TExecution$ must have a snapshot state for which the 
$\mathsf{Complete}$ and $\mathsf{NoConf}$ predicates hold.

The read committed isolation level is weaker than snapshot isolation. 
It guarantees that any $\mathsf{read}$ by a committed transaction will be from its own previous $\mathsf{write}$, 
and if no previous $\mathsf{write}$ exists, it will be from another committed transaction. 
Notice it does not guarantee that transactions read from a stable snapshot, just that 
reading happens from \emph{some} committed transaction.
To define the commit test of read committed, we use the predicates $\mathsf{ReadState}$ and $\mathsf{Preread}$
which we define in \Cref{def:predicates_rc}.
The $\TReadStates{\TExecution}{i}{k}{v}{\TState}$ predicate asserts that the state $\TState$ has an index $j$ 
in the execution $\TExecution$ which is smaller than the argument $i$, and that key $k$ points to value $v$ in the state.
The $\TPreread{(\TTransactionSet{\TTransaction}, \TTransactionOrder{\TTransaction})}{\TExecution}$ 
predicate says that all the $\mathsf{read}$ operations in a transaction 
$(\TTransactionSet{\TTransaction}, \TTransactionOrder{\TTransaction})$, with no 
prior writes to the same key, must have a read-state, as defined by the $\mathsf{ReadState}$ predicate, 
in the execution $\TExecution$. The commit test of read committed then 
states in \Cref{fig:state_based_levels} that all transactions in 
a valid execution $\TExecution$ must satisfy the $\mathsf{Preread}$ predicate.

\begin{definition}\label{def:predicates_rc}
  We define the predicates $\mathsf{ReadState}$ and $\mathsf{Preread}$ as follows:
  {
  \begin{align*}
    \intertext{\normalsize For an execution $\TExecution$, index $i$, key $k$, value $v$ and state $\TState$}
    \TReadStates{\TExecution}{i}{k}{v}{\TState} \eqdef{} & \exists j, j < i \land 
    \ListLookup{\TExecution}{j} = (\TState, \_) \land \ListLookup{\TState}{k} = v.
    \tag{ReadState} \label{def:read_states}
    \intertext{\normalsize For a transaction $(\TTransactionSet{\TTransaction}, \TTransactionOrder{\TTransaction})$ and execution $\TExecution$}
    \TPreread{(\TTransactionSet{\TTransaction}, \TTransactionOrder{\TTransaction})}{\TExecution} \eqdef{} & ~ \forall i ~ \mathit{id} ~ k ~ v, ~ 
    \ListLookup{\TExecution}{i} = (\_, (\TTransactionSet{\TTransaction}, \TTransactionOrder{\TTransaction})) \Rightarrow \tag{Preread} \label{def:preread}
    \TR(\mathit{id}, k, v) \in \TTransactionSet{\TTransaction} \Rightarrow \\
    & ~ \TNoPriorWrite{(\TTransactionSet{\TTransaction}, \TTransactionOrder{\TTransaction})}{\TR(\mathit{id}, k, v)} \Rightarrow
    \exists \TState, \TReadStates{\TExecution}{i}{k}{v}{\TState}.
  \end{align*}
  }
\end{definition}

\noindent The last isolation level, read uncommitted, can be understood as a baseline level in which we do not 
get any additional guarantees than our baseline transaction assumptions (transactions read from their 
most recent write). Thus, the commit test for read uncommitted is simply $\TRUE$.
The definitions of the read uncommitted and read committed levels are equivalent to the ones 
stated in the original formulation of the state-based model \citep{Crooks17}. The definition of the snapshot isolation level
has been changed as a result of removing the simplifying assumptions listed in the beginning of this section.
Specifically, the previous formulation of the $\mathsf{NoConf}$ predicate did not take into account that the same value can be written multiple 
times.
\begin{figure}[H]
  \vspace{-4mm}
  \small
  \begin{center}
    \begin{tabular}{ |c|c| } 
    \hline
    \textbf{Isolation Level} & \textbf{Commit Test} \\
    \hline
    Read Uncommitted \big($\TComTestRU{\TTransaction}{\TExecution}$\big) & 
    $\TRUE$ \\ 
    \hline
    Read Committed \big($\TComTestRC{\TTransaction}{\TExecution}$\big) & 
    $\TPreread{\TTransaction}{\TExecution}$ \\ 
    \hline
    Snapshot Isolation \big($\TComTestSI{\TTransaction}{\TExecution}$\big) & 
    $\exists \TState ~ i, ~ 
    \ListLookup{\TExecution}{i} = (\TState, \TTransaction) ~ \land$ 
    $\TComplete{\TExecution}{\TTransaction}{\TState}{i} ~ \land$
    $\TNoConf{\TExecution}{\TTransaction}{\TState}$ \\ 
    \hline
    \end{tabular}
  \end{center}
  \caption{Isolation levels defined using state-based commit tests.}
  \label{fig:state_based_levels}
\end{figure}
\paragraph{Correctness of implementation} In this paper, our goal is to show correctness of 
database implementations. The correctness criterion we use is a straightforward extension of 
the correctness criterion used by \citet{Crooks17} to prove equivalence between 
their state-based model and the dependency graph model.

\begin{definition}[Implementation Correctness]\label{def:implements}
  A database implements an isolation level $\TLevel$ when, for any set of 
  transactions $\TTransactionSetOf$ it produces, 
  there exists an execution $\TExecution$ such that \\ $\ExecOf{\ComTrans{\TTransactionSetOf}}{\TExecution}$ holds, and we have $\forall \TTransaction \in \ComTrans{\TTransactionSetOf}, ~ \TComTest{\TLevel}{\TTransaction}{\TExecution}$.
\end{definition}

\noindent We recognize that the word "produces" leaves some ambiguity in the definition.
We make this notion precise in \Cref{sec:4}.
Futher, a database in this paper is a program that implements the API presented in \Cref{sec:3.1}.
The takeaway from this section is that we have now defined correctness at the
transactional model level. 


%% file: sections/spec.tex
\section{Specifying Transactional Consistency in Separation Logic}
\label{sec:3}

In this section, we recall the formulation of transactional consistency for program executions in separation logic from \citet{SI-placeholder}.
The paper presents separation logic specifications for the three weak isolation levels read uncommitted, read committed 
and snapshot isolation. 
The works that we build upon
\citep{SI-placeholder,DBLP:conf/esop/Krogh-Jespersen20,free-theorems} are all
based on the Iris program logic framework \citep{irisjournal}, a state-of-the-art
framework for specifying and proving correctness of programs using separation
(Hoare) logic \citep{DBLP:journals/cacm/Hoare69,1029817}.
We will present the (minor) details of separation logic (and the Iris
framework specifically) that are necessary for understanding this paper on
per-need basis.

In this paper, we focus on snapshot isolation to explain our methodology for connecting abstract 
transactional consistency models to program executions. We emphasize again that our approach has been formalized for all the isolation levels 
in \citet{SI-placeholder}, and refer to the Rocq formalization for the details about the remaining isolation levels.
In \Cref{sec:3.1}, we explain the structure of the specification for weak isolation levels with a focus on parametricity of resources.
This is a property we rely on, in \Cref{sec:4} and \Cref{sec:5}, to achieve our connection to the state-based model with no 
observable changes in the specifications. 
After this, we go into detail with the per-operation specifications for snapshot isolation in \Cref{sec:3.2}
and explain how they expose linearization points through atomic specifications ---
another crucial property for establishing our connection to the state-based model.
For readers unfamiliar with separation logic, we show how the per-operation specifications can be used to prove safety of small examples, 
exemplified with the \emph{write skew} example, in \Cref{sec:example}.
Proving safety of small examples that capture interesting properties about isolation levels is, 
until the contribution of this paper, the best assurance we have that separation logic specifications 
capture isolation levels correctly.  

\subsection{Specification Structure}
\label{sec:3.1}

The snapshot isolation specification ($\Spec_{\mathit{SI}}$) is shown in \Cref{fig:specs_si}.
(Our presentation of the specification omits resources and specifications related to initialization of the database. 
These elements are not of importance to the proofs of this paper, because the correctness criteria (\Cref{def:implements}) 
 we prove is only about what is observable to clients of the database.) 
The specification for read uncommitted and read committed have the same structure. 
The specification takes as argument a transactional library $\mathit{lib}$. 
This is the client API of the database implementation with the method signatures 
$\mathsf{init} : \mathsf{ip} \rightarrow \mathsf{conn}$,
$\mathsf{start} : \mathsf{conn} \rightarrow \mathsf{unit}$,
$\mathsf{read} : \mathsf{conn} \times \mathsf{key} \rightarrow \mathsf{option ~
  val} $,
$\mathsf{write} : \mathsf{conn} \times \mathsf{key} \times \mathsf{val}
\rightarrow \mathsf{unit}$ and
$\mathsf{commit} : \mathsf{conn} \rightarrow \mathsf{bool}$, where
$\mathsf{conn}$ is the type for connections.
Any implementation verified against the specification $\Spec_{\mathit{SI}}$ must have these signatures, as this is what the 
specification assumes. 
The specification starts by stating the existence of a number of resources. 
These are all Iris predicates (as signified by the type $\ipropType$, which is the type of Iris propositions).
The first resource, $\ConnectionStateSymb: \SIconnType ~ \times ~ \statusType \rightarrow \ipropType$, 
is used to keep track of the state a client is in. A client can create this resource using the 
operation $\mathsf{init}$, and it will show  whether a client has an actively running transaction, 
on its connection of type $\SIconnType$, using a $\statusType$ type argument, which we will see examples of shortly.
The two resources, $\mapstoMem{\_}{\_}: \SIkeyType ~ \times ~ \histType \rightarrow \ipropType$ and 
$\mapstoCache{\_}{\_}{\_}: \SIconnType ~ \times ~ \SIkeyType ~ \times ~ \valType \rightarrow \ipropType$ are so called 
\emph{points-to resources}; they both have a $\SIkeyType$ arguments which points to either a history or a value. 
We call $\mapstoCache{\_}{\_}{\_}: \SIconnType ~ \times ~ \SIkeyType ~ \times ~ \valType \rightarrow \ipropType$ a \emph{local} points-to resource. 
It is used for describing the starting snapshot that a transaction reads from, thus it also takes a connection as argument.
We name $\mapstoMem{\_}{\_}: \SIkeyType ~ \times ~ \histType \rightarrow \ipropType$ the 
\emph{global} points-to resource. It describes the state of the database, where a history 
is a list of values holding all the values that has been written in the past. 
The resource $\KeyUpdStatusSymb: \SIconnType \times \SIkeyType \times \SIboolType \rightarrow \ipropType$
is used to keep track of whether a transaction made any updates to its local points-to resources.
This is essential information to have when resolving whether a transaction has any write conflicts
and safely can commit.
The specification $\Spec_{\mathit{SI}}$ provides $\Sep_{\Key\, \in\, \Keys} \SI\mapstoMem{\Key}{\HistEmpty}$
which are initial \emph{global} points-to resources representing the state in which no value has been written 
to any key yet. The symbol $\Sep$ is the iterated separating conjunction. 
The separating $\ast$ conjunction is much like an ordinary conjunction, except that any exclusively owned 
resources can only be used to satisfy one side of the conjunction, \ie it enforces a separation of resources. 

Following the listing of all resources comes the per-operation separation logic specifications in the form of Hoare triples
for each operation in the client API which the argument $\Spec_{\mathit{SI}}$ is an implementation of.
Concretely, we have a Hoare triple specification for the operations 
$\mathsf{init}$, $\mathsf{start}$, $\mathsf{read}$,
$\mathsf{write}$ and $\mathsf{commit}$. 
The operations $\mathsf{read}$, $\mathsf{write}$ and $\mathsf{commit}$, 
correspond to the ones we saw in our formalization of the state-based model in \Cref{sec:2}, 
while the $\mathsf{start}$ operation starts new transactions, and the $\mathsf{init}$ operation initializes a client connection with the transactional library.
\begin{figure}[htbp]
  \centering
  \vspace{-4mm}
  \begin{align*}
    \Spec_{\mathit{SI}}(\mathit{lib}) \eqdef ~ 
    \exists{} & ~ \ConnectionStateSymb: \SIconnType \times \statusType \rightarrow \ipropType, ~
    \mapstoMem{\_}{\_}: \SIkeyType \times \histType \rightarrow \ipropType, \\
    & \mapstoCache{\_}{\_}{\_}: \SIconnType \times \SIkeyType \times \valType \rightarrow \ipropType, ~
    \KeyUpdStatusSymb: \SIconnType \times \SIkeyType \times \SIboolType \rightarrow \ipropType. \\
    & \Sep_{\Key\, \in\, \Keys} \mapstoMem{\Key}{\HistEmpty} ~ \land 
    \text{\ruleref{si-init-client-spec} holds } \land
    \text{\ruleref{si-start-spec} holds } \land \\
    & \text{\ruleref{si-write-spec} holds } \land 
    \text{\ruleref{si-commit-spec} holds } \land
    \text{\ruleref{si-read-spec} holds}
  \end{align*}
  \caption{Specification for snapshot isolation.}
  \label{fig:specs_si}
  \vspace{-1mm}
\end{figure}
\paragraph{Parametric resources} 
As mentioned above, all resources 
in the specification are existentially quantified. 
All a user of the specification knows about the resources is that they work with the 
per-operation specifications. Thus, any client program, which is verified using $\Spec_{\mathit{SI}}$,
can work with any database implementation that implements the transactional library 
as specified in $\Spec_{\mathit{SI}}$. When our proofs in later sections of the paper 
modifies the resources of $\Spec_{\mathit{SI}}$, it is not observable to any client because we make 
sure the modified resources can still satisfy $\Spec_{\mathit{SI}}$, and due to the existential quantification, 
the client will never know that the underlying resources have changed. 
Further, we preemptively emphasize that our proofs, showing that a database implementation $\mathit{lib}$
is correct w.r.t. the state-based model, do not assume anything about the
implementation except that it implements specification ($\Spec_{\mathit{SI}}(\mathit{lib})$).

\subsection{Per-operation (Atomic) Specifications} 
\label{sec:3.2}

The per-operation specifications for all the database operations are shown in \Cref{fig:specs_si_operations}.
All operations, except the $\mathsf{init}$ operation, are not specified using ordinary Hoare triples but using so-called \emph{atomic} Hoare triples. 
To explain the difference between the two, we first recall that a Hoare triple $\anhoare{P}{\expr}{\Ret \val. Q}{}{}$ says that if the resources in $P$ are satisfied, 
the expression $\expr$ is safe to execute, and if $\expr$ terminates the resources in $Q$ will hold, with
the return value of $\expr$ bound to $\val$ in $Q$. With a Hoare triple $\anhoare{P}{\expr}{\Ret \val. Q}{}{}$, 
the resources in $P$ are consumed by the duration $\expr$ executes and \emph{only} when $\expr$ terminates, the resources 
in $Q$ are given back to the user of the specification. In highly concurrent systems where resources are shared using invariants, 
this is not strong enough. Thus, we specify operations using atomic Hoare triples $\anhoareatomic{P}{\expr}{\Ret \val. Q}{}{}$ \citep{10.1007/978-3-662-44202-9_9}. 
Intuitively, an atomic Hoare triple asserts that the expression being specified has a linearization point around which 
the resources $P$ are transformed to $Q$. Thus, around atomic Hoare triples, we are allowed to soundly open invariants 
as shown in \ruleref{Inv-atomic}.
\begin{figure}[htbp]
\vspace{-5mm}
  \begin{mathpar}
  \inferH{Inv-atomic}
  {\proves \anhoareatomic{R \ast P}{\expr}{\Ret \val. Q \ast R}{}{}}
  {\knowInv{}{R} \proves \anhoareatomic{P}{\expr}{\Ret \val. Q}{}{}}
  \end{mathpar}
\vspace{-8mm}
\end{figure}
An invariant holding a resource $R$ is written $\knowInv{}{R}$. 
From the exclusive ownership of a resource $R$, we can create an invariant $\knowInv{}{R}$ which does not assert exclusive ownership 
of $R$ any longer. Instead, $\knowInv{}{R}$ is \emph{duplicable} and can be freely shared.
Usually, invariants in Iris can only be opened around expressions that can step to a value in one step according to the operational semantics, 
atomic Hoare triples lets us lift this restriction.
Next, we go into details with each of the operation specifications. 

\paragraph{Init} Starting with the $\mathsf{init}$ operation, used for initializing a client connection, 
it takes no precondition and returns the $\CanStart{\Cst}$ resource. The resource \\ 
$\CanStart{\Cst}$ reflects 
that there is currently no active transaction on the connection $\Cst$, and that a new transaction can be started using the $\mathsf{start}$ operation.
\begin{figure}[h!]
  \vspace{-3mm}
  \centering
  \begin{mathpar}
    \axiomH{si-init-client-spec}
    { 
      \anhoare
        {
          \TRUE
        }
        {
          {\texttt{\scriptsize\linespread{0.4}init}} ~ 
          \srvsa
        } 
        { \begin{array}[t]{@{}l@{}}
          \Ret \Cst. \CanStart{\Cst}
          \end{array}
          }
        {}{}
    } 
    \and
    \axiomH{si-read-spec}
    { 
      \anhoareatomic
        { \begin{array}[t]{@{}l@{}}
          \mapstoCache{\Key}{\Cst}{\ValueOption}
          \end{array}
        }
        { \begin{array}[t]{@{}l@{}}
          {\texttt{\scriptsize\linespread{0.4}read}} ~ 
          \Cst ~ \Key
          \end{array}
        } 
        { \begin{array}[t]{@{}l@{}}
          \Ret \ValueOption. 
          \mapstoCache{\Key}{\Cst}{\ValueOption}
        \end{array}}
        {}{}
    } \\
    \and
    \vspace{3mm}
    \axiomH{si-write-spec}
    { 
      \anhoareatomic
        {
         \begin{array}[t]{@{}l@{}}
          \mapstoCache{\Key}{\Cst}{\ValueOption} \ast
          \KeyUpdStatus{\Cst}{\Key}{\Boolean}
        \end{array}
        }
        {
          {\texttt{\scriptsize\linespread{0.4}write}} ~ 
          \Cst ~ \Key ~ \val
        } 
        { \begin{array}[t]{@{}l@{}}
          \Ret \TT. \mapstoCache{\Key}{\Cst}{\Some \val} \ast 
                                      \KeyUpdStatus{\Cst}{\Key}{\TRUE}
          \end{array}
          }
        {}{}
    } \\
    \and
    \vspace{5mm}
    \axiomH{si-start-spec}
    { 
      \anhoareVatomic
        { \begin{array}[t]{@{}l@{}}
          \CanStart{\Cst} \ast
          \Sep_{(\Key,\, \Hist) \in\, \map} \mapstoMem{\Key}{\Hist}
        \end{array}
        }
        { \begin{array}[t]{@{}l@{}}
          {\texttt{\scriptsize\linespread{0.4}start}} ~ 
          \Cst
        \end{array}
        } 
        { \begin{array}[t]{@{}l@{}}
          \Ret \TT. \Active{\Cst}{\map} \asts \\
          \Sep_{(\Key,\, \Hist) \in\, \map} \Big(\mapstoMem{\Key}{\Hist} \asts
          \mapstoCache{\Key}{\Cst}{\HistVal{\Hist}} \ast 
          \KeyUpdStatus{\Cst}{\Key}{\FALSE} \Big)
        \end{array}}
        {}{}
    }
    \and
    \axiomH{si-commit-spec}
    { 
      \anhoareVatomic
        { \begin{array}[t]{@{}l@{}}
          \Active{\Cst}{\Snapshot} \ast \dom ~ \map = \dom ~ \Snapshot =  \dom ~ \Cache \asts \\[0.4em]
          \Sep_{(\Key,\, \Hist) \in\, \map} \mapstoMem{\Key}{\Hist} ~ \asts
          \Sep_{(\Key,\, (\ValueOption,\, \Boolean)) \in\, \Cache} 
            \big(\mapstoCache{\Key}{\Cst}{\ValueOption} \ast \KeyUpdStatus{\Cst}{\Key}{\Boolean}\big)
          \end{array}
        }
        { \begin{array}[t]{@{}l@{}}
          {\texttt{\scriptsize\linespread{0.4}commit}} ~ 
          \Cst
          \end{array}
        }
        { \begin{array}[t]{@{}l@{}}
          \Ret \val. \CanStart{\Cst} ~ \ast ~ \\[0.4em]
          \Big(\val = \TRUE \ast \CanCommitPredicate{\map}{\Snapshot}{\Cache} \ast 
          \Sep_{\substack{(\Key,\, \Hist) \in\, \map \\ 
          (\Key,\, p) \in\, \Cache}} \mapstoMem{\Key}{\CommitHist{p}{\Hist}}\Big) ~ \lor ~ \\[0.7cm]
          \Big(\val = \FALSE \ast \lnot \CanCommitPredicate{\map}{\Snapshot}{\Cache} ~ \asts
          \Sep_{(\Key,\, \Hist) \in\, \map} \mapstoMem{\Key}{\Hist}\Big)
        \end{array}
        }
        {}{}
    }
    \end{mathpar}
    \caption{Per-operation specification for snapshot isolation.}
    \label{fig:specs_si_operations}
\end{figure}

\paragraph{Start}
The $\mathsf{start}$ operation is specified using an atomic Hoare triple. This is because 
it takes as precondition the global points-to resource. 
As potentially many concurrent transactions will be starting and stopping, 
the database state, in the form of the global points-to resources, will typically reside in an invariant,
shared among all clients of the database.
The global points-to resources are described by a map $m$, 
which is the current snapshot of the database. 
The $\mathsf{start}$ operation takes the resource $\CanStart{\Cst}$ and 
changes its state to active for the current snapshot $m$. 
In the postcondition, global points-to resources are returned 
unchanged and can be put back into a shared invariant.
Further, local points-to resources are created for the client 
to use, based on the snapshot $m$,
together with $\KeyUpdStatusSymb$ resources that reflects no writes have been done yet.

\paragraph{Write and read} The specifications for the $\mathsf{write}$ and $\mathsf{read}$ operations are almost identical to specifications 
for load and store operations in sequential memory; 
changes are immediately propagated and observed in the local points-to resources. 
Apart from updating the local points-to resources, the write specification also marks,
using the $\KeyUpdStatusSymb$ resource, that the current transaction has written to the key $\Key$. 
We use this information in the commit specification to detect write conflicts with concurrent transactions and 
decide whether the current transaction is allowed to commit.

\paragraph{Commit} A client calls the $\mathsf{commit}$ operation with the resource $\Active{\Cst}{\Snapshot}$ 
stating that the current transaction has been reading from the starting snapshot $ms$. 
The precondition also takes the global points-to resources, recorded by the current snapshot of the database $m$, 
and the changes of the client, in the form of local points-to and $\KeyUpdStatusSymb$ resources, described by the map $mc$.
No matter the outcome of the $\mathsf{commit}$ operation, the postcondition includes $\CanStart{\Cst}$
to reflect that the client has ended its transaction, and it can start a new one.
The rest of the postcondition is a disjunction representing the successful and the unsuccessful commit case. 
To capture write conflicts with concurrent transactions, we use the $\CanCommitPredicateSymb$ 
predicate in \eqref{fig:commit_pred}.
\begin{align}
 \CanCommitPredicate{\map}{\Snapshot}{\Cache} \eqdef
      \forall \Key \in \Keys, ~ \MapLookup{\Cache}{\Key} = \Some{(p, \TRUE)}
      \Rightarrow \MapLookup{\map}{\Key} = \MapLookup{\Snapshot}{\Key}. \label{fig:commit_pred}
\end{align}
\noindent The predicate is true if for each key the current transaction has written to (the keys 
where $\MapLookup{\Cache}{\Key} = \Some{(p, \TRUE)}$ holds as $\TRUE$ comes from the $\KeyUpdStatusSymb$ resource 
which is updated in the write specification) the transaction's starting snapshot $ms$ is equal to the current snapshot 
of the database $m$. The database state is tracked with the global points-to resources using histories, 
\ie list of values representing all past writes,
to be able to check for changes between snapshots of the database. 
Using values instead of histories is not sound when comparing different snapshots where the latest value is the same. 
In the successful commit case of the postcondition, 
the global points-to resources are updated with the writes made by the committing transaction
by appending to the current histories. 
The global points-to resources are left unchanged in the unsuccessful commit case. 


%% file: sections/method.tex
\section{Connecting Transactional Models to Semantic Models of Program Execution}
\label{sec:4}
In this section, we explain our method for extracting transactions
from database implementations satisfying the Hoare triple specifications of \citet{SI-placeholder}.
Ultimately, we wish to prove that a database implementation 
for an isolation level $\TLevel$ correctly implements that isolation level as per 
the definition of implementation correctness from \Cref{def:implements}.
As the definition of implementation correctness assumes a set of transactions
as defined in the state-based model (\Cref{def:transaction}),
the transactions we extract must conform to this definition, 
implying that transactions must read their own writes.
The extraction, from code to transactions usable in the state-based model, 
is depicted in \Cref{fig:extraction_method}.
Before going into detail with each level of the extraction in the coming subsections, 
we give a more detailed overview of the extraction than we did in \Cref{sec:overview}.

The extraction works across four different levels, with the first level (level 0) being the executable code and the last level (level 3)
being an extracted set of transactions $\TTransactionSetOf$ for which 
implementation correctness must hold, namely
there must exist an execution at the model level that satisfies a given commit test. 
The core idea here is that we \emph{instrument} the database code to produce a trace of \emph{events} as the database runs. 
The instrumentation happens at level 0.
(We emphasize that this instrumentation bears no effect on the execution of the program; it is \emph{ghost code} \citep{10.1145/360051.360224} which one can erase before compiling and running the code.)
We instrument each operation so that it produces a so-called \emph{pre event} and a \emph{post event}, marking the start and end of the operation --- these are tagged by a globally unique identifier for the operation in question.
The trace of pre and post events constitute level 1; events emitted from any operation 
all go into the same trace.  
Throughout the extraction, for each operation we also produce its linearization point event 
(also tagged by the operation's unique identifier).
Unlike the pre and post events, the linearization point events are not produced by code instrumentation because they are only determined at the logical level, taking into account all other concurrently running operations from other transactions.
Thus, linearization point events are captured as a separate sequence of events which need to be suitably coherent with the trace of pre and post events (this will be made precise in the coming sections).
Intuitively, the order of the linearization points is the order in which operations are observed by the database.
The trace of linearization points makes up level 2 in the extraction.
From the trace of linearization points, we have enough information to extract a set of transactions which 
makes up level 3. A single linearization point corresponds to a single operation, 
and as the trace holds the order in which the points are captured, 
we have enough information to soundly extract transactions. 
Once the set of transactions have been extracted, we have all the information we need 
to reason in the state-based model; implementation correctness \Cref{def:implements}
says that for any set of transactions a database produces, 
there must exist an execution for which the commit test of the isolation level the database 
implements must hold --- our extraction method exactly 
captures the transactions produced by the database.

For the rest of this paper, we assume that we are given a database
implementation (for each isolation level) that satisfies the specifications given
by \citet{SI-placeholder}. Note that we make no further assumptions on the implementation,
nor on how the specifications have been proven;
indeed, we show that \emph{any} implementation that satisfies those
specifications correctly implements the relevant isolation level in the sense of
\Cref{def:implements}. 

Next, we explain each level in detail, and use the \emph{write skew} example from \Cref{fig:SI_skew_example} as our running example. 

\subsection{Executable Code Instrumentation (Levels 0 and 1)}
\label{sec:4.1}
At level 0, we only instrument the client facing interface of the database implementation to emit pre and post events. 
Note that \Cref{fig:extraction_method} only shows the instrumentation for the
$\mathsf{read}$ operation; all the operations of the interface are instrumented similarly
as we display in \Cref{fig:instrumentation}.
\begin{figure}[!h]
  \vspace{-5mm}
  \centering
  \begin{align*}
  \Wrap{{\texttt{\scriptsize\linespread{0.4}start}} ~ \Cst} & \eqdef \SIletin{\tau}{\textlang{\scriptsize\color{ACMDarkBlue}fresh} ~ (\Cst, \mathbf{S})} ~ {\texttt{\scriptsize\linespread{0.4}start}} ~ \Cst; ~ \textlang{\scriptsize\color{ACMDarkBlue}emit} ~ (\tau, \Cst, \mathbf{S}); () \\
  \Wrap{{\texttt{\scriptsize\linespread{0.4}read}} ~ \Cst ~ \Key} & \eqdef \SIletin{\tau}{\textlang{\scriptsize\color{ACMDarkBlue}fresh} ~ (\Cst, \mathbf{R}, k)} ~ \SIletin{ov}{{\texttt{\scriptsize\linespread{0.4}read}} ~ \Cst ~ \Key} ~ \textlang{\scriptsize\color{ACMDarkBlue}emit} ~ (\tau, \Cst, \mathbf{R}, k, ov); ov \\
  \Wrap{{\texttt{\scriptsize\linespread{0.4}write}} ~ \Cst ~ \Key ~ \val} & \eqdef \SIletin{\tau}{\textlang{\scriptsize\color{ACMDarkBlue}fresh} ~ (\Cst, \mathbf{W}, k, \val)} ~ {\texttt{\scriptsize\linespread{0.4}write}} ~ \Cst ~ \Key ~ \val; ~ \textlang{\scriptsize\color{ACMDarkBlue}emit} ~ (\tau, \Cst, \mathbf{W}, k); () \\
  \Wrap{{\texttt{\scriptsize\linespread{0.4}commit}} ~ \Cst} & \eqdef \SIletin{\tau}{\textlang{\scriptsize\color{ACMDarkBlue}fresh} ~ (\Cst, \mathbf{C})} ~ \SIletin{b}{{\texttt{\scriptsize\linespread{0.4}commit}} ~ \Cst} ~ \textlang{\scriptsize\color{ACMDarkBlue}emit} ~ (\tau, \Cst, \mathbf{C}, b); b
  \end{align*}
  \caption{Instrumentation of database operations} 
  \label{fig:instrumentation}
\end{figure}
In \Cref{fig:skew-lvl-1}, we show one possible trace of pre and post emit events that can come from the \emph{write skew} example 
from \Cref{fig:SI_skew_example}. 
In this particular trace (notice there are several possible traces that the \emph{write skew} example could produce), 
we display only the events for the $\mathsf{read}$ and $\mathsf{write}$ operations. 
To explain the structure of the events, let us consider the first event $(\tau_1, c_1, \mathbf{R}, x)$.
It consists of the tag $\tau_1$, the client connection $c_1$, a label $\mathbf{R}$, 
stating that the event is from a $\mathsf{read}$ operation, and the key being read $x$.
Thus, this is the pre event for the $\mathsf{read}$ operation of transaction $\TTransaction_1$ 
that the instrumentation has emitted.
The next event in the trace is the pre event for the $\mathsf{read}$ operation of transaction $\TTransaction_2$.
It is not until the third slot of the trace that the post emit event 
$(\tau_1, c_1, \mathbf{R}, x_0)$ of the $\mathsf{read}$ operation from $\TTransaction_1$ has been emitted.
The post event $(\tau_1, c_1, \mathbf{R}, x_0)$ is similar to the corresponding pre event $(\tau_1, c_1, \mathbf{R}, x)$:
the tag and connection are the same, tying the two events together, and they are marked as coming from a $\mathsf{read}$ operation using $\mathbf{R}$, 
but this time the value being read $x_0$ is returned.
\ie the result of reading the key $x$ is the value $x_0$ (the initial version of $x$). 
The interleaving of events show that $\mathsf{read}$ operations are concurrent. 
The later half of the trace display the pre and post emit events 
of the $\mathsf{write}$ operations from the \emph{write skew} example.
They too are interleaved and follow the instrumentation structure from \Cref{fig:instrumentation}.

\begin{figure}[h]
  \vspace{-2mm}
  \footnotesize
  \centering
  \begin{tikzpicture}[>=latex, scale=0.848, every node/.style={scale=0.848}]
    \matrix[mymats=white,anchor=north, text width=18.5mm, text height=3mm, text badly centered]
    at (0,0) 
    (mat1)
    {
      \langle \tau_1, c_1, \mathbf{R}, x \rangle & \langle \tau_2, c_2, \mathbf{R}, y \rangle & \langle \tau_1, c_1, \mathbf{R}, x_0 \rangle & \langle \tau_2, c_2, \mathbf{R}, y_0 \rangle & 
      \langle \tau_3, c_1, \mathbf{W}, x, x_1 \rangle & \langle \tau_4, c_2, \mathbf{W}, y, y_1 \rangle & \langle \tau_3, c_1, \mathbf{W} \rangle & \langle \tau_4, c_2, \mathbf{W} \rangle \\
    };
    \draw[black!70] (-7.1,-1) node[anchor = north] {\begin{minipage}{7em}\centering "Pre event of $\TTransaction_1$ \\ ~~ $\mathsf{read}$ operation" \end{minipage}};
    \draw[black!70] (-5,-1) node[anchor = north] {\begin{minipage}{7em}\centering "Pre event of $\TTransaction_2$ \\ ~~ $\mathsf{read}$ operation" \end{minipage}};
    \draw[black!70] (-3,-1) node[anchor = north] {\begin{minipage}{7em}\centering "Post event of $\TTransaction_1$ \\ ~~ $\mathsf{read}$ operation" \end{minipage}};
    \draw[black!70] (-0.9,-1) node[anchor = north] {\begin{minipage}{7em}\centering "Post event of $\TTransaction_2$ \\ ~~ $\mathsf{read}$ operation" \end{minipage}};
    \draw[black!70] (1.1,-1) node[anchor = north] {\begin{minipage}{7em}\centering "Pre event of $\TTransaction_1$ \\ ~~ $\mathsf{write}$ operation" \end{minipage}};
    \draw[black!70] (3.1,-1) node[anchor = north] {\begin{minipage}{7em}\centering "Pre event of $\TTransaction_2$ \\ ~~ $\mathsf{write}$ operation" \end{minipage}};
    \draw[black!70] (5.2,-1) node[anchor = north] {\begin{minipage}{7em}\centering "Post event of $\TTransaction_1$ \\ ~~ $\mathsf{write}$ operation" \end{minipage}};
    \draw[black!70] (7.2,-1) node[anchor = north] {\begin{minipage}{7em}\centering "Post event of $\TTransaction_2$ \\ ~~ $\mathsf{write}$ operation" \end{minipage}};
  \end{tikzpicture}
  \caption{Extraction of the \emph{write skew} example at level 1 (events from the $\mathsf{start}$ and $\mathsf{commit}$ operations are omitted).} 
  \label{fig:skew-lvl-1}
\end{figure}

In terms of the operational semantics, we have integrated the code instrumentation primitives $\Tfresh$ and $\Temit$ of \citet{free-theorems} into the operational semantics of AnerisLang, 
the programming language of the Aneris program logic (an instantiation of Iris) \citep{DBLP:conf/esop/Krogh-Jespersen20} used by \citet{SI-placeholder} in their formalization.
An excerpt of the operational semantics of AnerisLang (with the added code instrumentation primitives) is given in \Cref{fig:ghost_code_op}.
See \citet{free-theorems} and the Aneris documentation \citep{aneris-documentation} for details.%
\footnote{Our presentation in \Cref{fig:ghost_code_op} is slightly simplified compared to the Aneris documentation \citep{aneris-documentation} where the step relation $\nhstep$ is split into several layers and only the network layer is explicitly marked with the node name.}
\begin{figure}[h]
  \vspace{-2mm}
  \raggedright
  \fbox{\tiny$(e, \sigma, t) \nhstep_{(n)} (e', \sigma', t')$}
  \begin{mathpar}
    \axiomH{Pair-Proj-Step1}
    {(\textlang{\scriptsize\color{ACMDarkBlue}fst} ~(\val_1, \val_2), \sigma, t) \nhstep_{(n)} (\val_1, \sigma, t)}
    \and
    \axiomH{Function-Call-Step}
    {((\RecV f (\var)= \expr) ~ \val, \sigma, t) \nhstep_{(n)} (e[(\RecV f (\var)= \expr)/f][\val/x], \sigma, t)}
    \and
    \inferH{Load-From-Heap-Step}
    {\sigma.\mathit{heap}_n(\ell) = v}
    {(!\ell, \sigma, t) \nhstep_{(n)} (v, \sigma, t)}
    \and
    \axiomH{Emit-Step}
    {(\textlang{\scriptsize\color{ACMDarkBlue}emit} ~ \val, \sigma, t) \nhstep_{(n)} ((), \sigma, t \cdot \val)}
    \and
    \inferH{Fresh-Step}
    {\tau ~\text{ is fresh in trace }~t}
    {(\textlang{\scriptsize\color{ACMDarkBlue}fresh} ~ \val, \sigma, t) \nhstep_{(n)} (\tau, \sigma, t \cdot \langle\tau, \val\rangle)}
  \end{mathpar}
  \caption{An excerpt of Aneris's operational semantics; $\sigma$ tracks the state of heaps of all nodes, and the global state of the network; $t$ is the global trace of all events; $n$ is the name of the node in the network (ip address).} 
  \label{fig:ghost_code_op}
  \vspace{1mm}
\end{figure}
Note that in the rule \ruleref{Emit-Step}, the value $\val$ is added to the global trace $t$. Likewise, 
the fresh expression also appends its argument value $\val$ to the global trace, but it does so together with a freshly generated tag $\tau$ --- see \Cref{fig:extraction_method} for an example.

The core idea of the methodology of \citet{free-theorems} is to 
enforce a so-called trace invariant, written $\TtraceInv{\mathcal{L}}$, which
guarantees that the global trace of events must belong to the language
$\mathcal{L}$ (a set of traces). In the terminology of the Iris framework,
$\TtraceInv{\mathcal{L}}$ is an example of an \emph{invariant}, a property that all verified
programs, regardless of their pre- and postconditions must uphold at all times.
This is reflected in the proof rules \ruleref{Ht-emit} and
\ruleref{Ht-fresh} for the ghost operations $\SIemC{}$ and $\SIfrC{}$ in \Cref{fig:ghost_code_ht}. Here, $\Ttrace{t}$ is a
predicate stating that the global trace is currently $t$; taken together with
$\TtraceInv{\mathcal{L}}$, one concludes that $t \in \mathcal{L}$. Note how
these rules ensure that the trace invariant is preserved by the execution steps of
the code instrumentation. 
\begin{figure}[h]
  \vspace{-2mm}
  \centering
  \begin{mathpar}
    \inferH{Ht-emit}
    {t \cdot v \in \mathcal{L}}
    {\TtraceInv{\mathcal{L}}
     \vdash
     \anhoare
      {
        \begin{array}[t]{@{}l@{}}
          \Ttrace{t}
        \end{array}
      }
      {\textlang{\scriptsize\color{ACMDarkBlue}emit} ~ \val}
      {
        \begin{array}[t]{@{}l@{}}
          \Ttrace{t \cdot v}
        \end{array}
      }
      {}{}
    } 
    \and
    \inferH{Ht-fresh}
    {\forall \tau, ~ \tau ~\text{ is fresh in }~ t \Rightarrow t \cdot \langle\tau, \val\rangle \in \mathcal{L}}
    {\TtraceInv{\mathcal{L}}
     \vdash
     \anhoare
      {
        \begin{array}[t]{@{}l@{}}
          \Ttrace{t}
        \end{array}
      }
      {\textlang{\scriptsize\color{ACMDarkBlue}fresh} ~ \val}
      {
        \Ret \tau.
        \begin{array}[t]{@{}l@{}}
          \tau \not\in \TTags(t) \land \Ttrace{t \cdot \langle\tau, \val\rangle}
        \end{array}
      }
      {}{}
    }  
  \end{mathpar}
  \caption{Hoare-triple specifications for the $\Tfresh$ and $\Temit$ code instrumentation operations.} 
  \label{fig:ghost_code_ht}
\end{figure}
It is essentially our definition of $\mathcal{L}$ that makes up the extraction of \Cref{fig:extraction_method}.
We choose $\mathcal{L} \eqdef \{\Trace ~ | ~ \ValTrace{\TLevel}{\Trace} \}$ where the predicate $\ValTrace{\TLevel}{\Trace}$, 
which we refer to as the global proof invariant, is defined in \Cref{fig:global_invariant}.
The argument $\TLevel$ is the isolation level that the database is implementing, and 
$\Trace$ is the trace of pre and post events created by the instrumentation.
The global proof invariant asserts the existence of a sequence of linearization points $\LTrace$ (level 2), 
a set of transactions $\TTransactionSetOf$ (level 3) and an execution $\TExecution$. 
On the first line it states validity of $\LTrace$ and $\TTransactionSetOf$ 
(we will explain these predicates below),
before asserting implementation correctness, as in \Cref{def:implements}, in the first half of the second line.
The last three predicates makes sure there is coherence between the different 
levels (we will explain these predicates below), \ie that a given level is consistent with its previous level.
{
\begin{align}
  \ValTrace{\TLevel}{\Trace} \eqdef{} & \exists ~ \LTrace ~ \TTransactionSetOf ~ \TExecution, ~ 
  \ValSeq{\LTrace} \land \ValTrans{\TTransactionSetOf} ~ \land \tag{The Global Proof Invariant} \label{fig:global_invariant} \\
  & (\forall \TTransaction \in \ComTrans{\TTransactionSetOf}. ~ \TComTest{\TLevel}{\TTransaction}{\TExecution}) \land 
  \CohTrace{\Trace}{\LTrace} ~ \land \notag \\ 
  & \CohTrans{\LTrace}{\TTransactionSetOf} \land \ExecOf{\ComTrans{\TTransactionSetOf}}{\TExecution} \notag
\end{align}
}
We remark that because our extraction method is expressed as an invariant,
we capture \emph{all} the possible transactions a database implementation produces
under \emph{any} client interaction, and show that these are sound transactions. 

\subsection{Linearization Points (Level 2)}

At level 2, we have a trace of linearization points. A single linearization
point must correspond to a pair of pre and post emit events from the global
trace at level 1, as captured by the predicate $\CohTrace{\Trace}{\LTrace}$, which
asserts that:
\begin{enumerate}
  \item The order of $\Trace$ is respected in $\LTrace$: for any $\mathit{lp}_1$ that comes before a $\mathit{lp}_2$ in $\LTrace$, the post event of $\mathit{lp}_2$ does not come before the pre event of $\mathit{lp}_1$ in $\Trace$.
  \item Every pair of pre and post events in $\Trace$ has a corresponding
        linearization point (agreeing on the tag, arguments, and return value of the operation).
        \label{item:pre-pos-must-have-lin}
  \item Every linearization point in $\LTrace$ must have a pre event in $\Trace$.
  \item Linearization points in $\LTrace$ are unique.
\end{enumerate}
Note that the predicate $\CohTrace{\Trace}{\LTrace}$ essentially captures the standard
linearizability correctness criteria for concurrent objects
\citep{linearizability+strict-serializability};
there is a single point, between the invocation and end of an operation, 
where the effect of the operation appears instantaneously. 
On the logical level, linearization points in separation logic are exposed through 
atomic Hoare triples. Recall the per-operation specifications from \Cref{sec:3.2}. 
They were specified using atomic Hoare triples (such that clients can use
invariants to reason about concurrent threads, and invariants can be
opened around the database operations).
As the extraction we are defining in this section is a trace invariant,
\ie an invariant in Iris, we can use the fact that the operations are 
specified with atomic Hoare triples to insert 
linearization points in the sequence at level 2, 
in between the emit of the pre event and the post event in the instrumented code.
The insertion happens as part of the proof that from verified database implementations we extract valid transactions  
w.r.t. the state-based model, we will see this in detail in \Cref{sec:5}. Because linearization points are
unique, and we define them to hold all information about the operations they correspond to, 
we can essentially treat them as operations. 
The predicate $\ValSeq{\LTrace}$ asserts that operations (their linearization
points) must appear in the expected order, \ie $\mathsf{start}$ must be called before
read and writes can be performed, $\mathsf{commit}$ can only appear after $\mathsf{start}$ (and
possibly some reads and/or  writes after $\mathsf{start}$), etc.

In \Cref{fig:skew-lvl-2}, we show one of the possible extractions of linearization points
for the trace of pre and post events from the \emph{write skew} example in \Cref{fig:skew-lvl-1}.
Notice they way we have defined linearization points: they correspond exactly 
to the definition of transaction operations we defined in the state-based model \Cref{sec:2}.
This sequence of linearization points is a sound extraction, as it satisfies the conditions of the $\ValSeqSym$ predicate. 
Noticeable, each pair of pre and post events have one linearization point, 
and the linearization points respect the order imposed by the trace pre and post events on the previous level;
as the $\mathsf{read}$ operation from the two transactions of the example are concurrent (their pre and post events overlap), 
their linearization points can be ordered either way, similarly for the $\mathsf{write}$ operations. 
But, it would be an unsound order of linearization points if one of the two $\mathsf{read}$ operations were after 
a linearization point from one of the two $\mathsf{write}$ operations, because 
the post event of any $\mathsf{read}$ operation comes before the pre event of any $\mathsf{write}$ operation. 
Thus, we know for sure that the $\mathsf{read}$ operations ended before the $\mathsf{write}$ operations started. 
\begin{figure}[h]
  \vspace{-2mm}
  \footnotesize
  \centering
  \begin{tikzpicture}[>=latex, scale=0.8, every node/.style={scale=0.79}]
    \matrix[mymats=white,anchor=north, text width=20.0mm, text height=3mm, text badly centered]
     at (0,1.5) 
    (mat2)
    {
      \TR\big((\tau_1, c_1), x, x_0\big) & \TR\big((\tau_2, c_2), y, y_0\big) & \TW\big((\tau_3, c_1), y, y_1\big) & \TW\big((\tau_4, c_2), x, x_1\big) \\
    };
    \matrix[mymats=white,anchor=north, text width=18.5mm, text height=3mm, text badly centered]
    at (0,0) 
    (mat1)
    {
      \langle \tau_1, c_1, \mathbf{R}, x \rangle & \langle \tau_2, c_2, \mathbf{R}, y \rangle & \langle \tau_1, c_1, \mathbf{R}, x_0 \rangle & \langle \tau_2, c_2, \mathbf{R}, y_0 \rangle & 
      \langle \tau_3, c_1, \mathbf{W}, x, x_1 \rangle & \langle \tau_4, c_2, \mathbf{W}, y, y_1 \rangle & \langle \tau_3, c_1, \mathbf{W} \rangle & \langle \tau_4, c_2, \mathbf{W} \rangle \\
    };
    \draw[black!70] (-7.1,-1) node[anchor = north] {\begin{minipage}{7em}\centering "Pre event of $\TTransaction_1$ \\ ~~ $\mathsf{read}$ operation" \end{minipage}};
    \draw[black!70] (-5,-1) node[anchor = north] {\begin{minipage}{7em}\centering "Pre event of $\TTransaction_2$ \\ ~~ $\mathsf{read}$ operation" \end{minipage}};
    \draw[black!70] (-3,-1) node[anchor = north] {\begin{minipage}{7em}\centering "Post event of $\TTransaction_1$ \\ ~~ $\mathsf{read}$ operation" \end{minipage}};
    \draw[black!70] (-0.9,-1) node[anchor = north] {\begin{minipage}{7em}\centering "Post event of $\TTransaction_2$ \\ ~~ $\mathsf{read}$ operation" \end{minipage}};
    \draw[black!70] (1.1,-1) node[anchor = north] {\begin{minipage}{7em}\centering "Pre event of $\TTransaction_1$ \\ ~~ $\mathsf{write}$ operation" \end{minipage}};
    \draw[black!70] (3.1,-1) node[anchor = north] {\begin{minipage}{7em}\centering "Pre event of $\TTransaction_2$ \\ ~~ $\mathsf{write}$ operation" \end{minipage}};
    \draw[black!70] (5.2,-1) node[anchor = north] {\begin{minipage}{7em}\centering "Post event of $\TTransaction_1$ \\ ~~ $\mathsf{write}$ operation" \end{minipage}};
    \draw[black!70] (7.2,-1) node[anchor = north] {\begin{minipage}{7em}\centering "Post event of $\TTransaction_2$ \\ ~~ $\mathsf{write}$ operation" \end{minipage}};
    \draw[*->] (mat1-1-1.north) -- (mat2-1-1.south);
    \draw[*->] (mat1-1-3.north) -- (mat2-1-1.south);
    \draw[*->] (mat1-1-2.north) -- (mat2-1-2.south);
    \draw[*->] (mat1-1-4.north) -- (mat2-1-2.south);
    \draw[*->] (mat1-1-5.north) -- (mat2-1-3.south);
    \draw[*->] (mat1-1-7.north) -- (mat2-1-3.south);
    \draw[*->] (mat1-1-6.north) -- (mat2-1-4.south);
    \draw[*->] (mat1-1-8.north) -- (mat2-1-4.south);
  \end{tikzpicture}
  \caption{Extraction of the \emph{write skew} example from level 1 to level 2 ($\mathsf{start}$ and $\mathsf{commit}$ operations are omitted).}
  \label{fig:skew-lvl-2}
\end{figure}
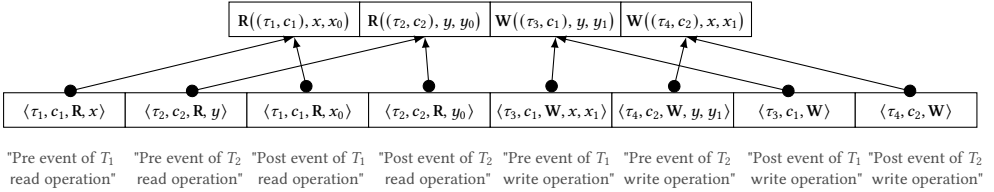
\vspace{-3mm}
\subsection{Set of Transactions Adhering to the State-based Model (Level 3)}
At level 3, we "chop" the sequence of operations (their linearization points)
into sets of transactions. The predicate
$\CohTrans{\LTrace}{\TTransactionSetOf}$ in \eqref{fig:global_invariant} asserts
coherence by stating that transactions in $\TTransactionSetOf$ have the same
operations as the trace $\LTrace$, any non $\mathsf{commit}$ operation must be in the same
transaction as the $\mathsf{commit}$ operation that follows in $\LTrace$ from the same
connection, and that for any two operations from a transaction $\TTransaction$
in $\TTransactionSetOf$ the program order of $\TTransaction$ must agree with the
ordering of the two operations in $\LTrace$. The predicate
$\ValTrans{\TTransactionSetOf}$ ensures that the set of transactions
$\TTransactionSetOf$ we extract meets the conditions expected by the state-based
model:

\begin{enumerate}
  \item There is at most one open transaction per connection. 
        An open transaction is a transaction which is started but not yet committed.
  \item No two transactions contain the same operation. (This includes the operation identifier/tag.)
  \item For each transaction we have: 
        \begin{itemize}
          \item All operations come from the same connection.
          \item A read of key $k$ reads the most recent write to $k$ (if any) in the same transaction. 
          \item There is at most one $\mathsf{commit}$ and it is the last operation of the transaction.
        \end{itemize}
\end{enumerate}

In \Cref{fig:skew-lvl-3}, we show the extraction of transactions from the sequence of linearization points 
in the \emph{write skew} example. In contrast to the previous figure, we have included $\mathsf{start}$ and $\mathsf{commit}$ linearization points.
Intuitively, for each connection, we create a transaction by 
going through the sequence of linearization points and add operations of the connection 
as they appear in the sequence. Notice how $\mathsf{start}$ operations are in the sequence of linearization points, 
but we do not use them when creating transactions. This reflects that $\mathsf{start}$ operations are used in implementations, 
but they are not included in the state-based model.

Because the extraction is an invariant, we sometimes have partially completed transactions in 
the set of transactions. For partially completed transactions, we assert they read their most recent 
write, and are otherwise well-formed as defined above, but we only impose that they should be included
in the reasoning at the state-based model level, to show implementation correctness, when they are complete.
\begin{figure}[h]
  \vspace{-2mm}
  \small
  \centering
  \begin{tikzpicture}[>=latex, scale=0.8, every node/.style={scale=0.8}]
    \matrix[mymats=white,anchor=north, text width=21.0mm, text height=3mm,]
    at (0,0) 
    (mat1)
    {
      \TR\big((\tau_1, c_1), x, x_0\big) & \TR\big((\tau_2, c_2), y, y_0\big) & \TW\big((\tau_3, c_1), y, y_1\big) & \TW\big((\tau_4, c_2), x, x_1\big)
      & \TC\big((\tau_5, c_1), \TRUE\big) & \TC\big((\tau_6, c_2), \TRUE\big) \\
    };
    \node[above=25pt of mat1]
    (mat3) {$\TTransactionSetOf = \{[\TR\big((\tau_1, c_1), x, x_0\big);\TW\big((\tau_3, c_1), y, y_1\big);\TC\big((\tau_5, c_1), \TRUE\big)],
            [\TR\big((\tau_2, c_2), y, x_0\big);\TW\big((\tau_4, c_2), x, x_1\big);\TC\big((\tau_6, c_2), \TRUE\big)] \}$};
    \draw[-{Implies},double, line width=0.25mm] (0,0.2) -- (0,0.9);
  \end{tikzpicture}
  \caption{Extraction of the \emph{write skew} example from level 2 to level 3.} 
  \label{fig:skew-lvl-3}
  \vspace{-2mm}
\end{figure}
\subsection{Implementation Correctness in the State-based Model}
The rest of \eqref{fig:global_invariant} asserts that there exists an execution $\TExecution$ that is coherent with the set of committed transactions
($\ComTrans{\TTransactionSetOf}$ is the subset of committed transactions in $\TTransactionSetOf$), 
\ie $\ExecOf{\ComTrans{\TTransactionSetOf}}{\TExecution}$ as defined in \Cref{sec:2}, 
and for which implementation correctness holds, \ie
$(\forall \TTransaction \in \ComTrans{\TTransactionSetOf}. ~ \TComTest{\TLevel}{\TTransaction}{\TExecution})$.

For the transactions extracted in the \emph{write skew} example \Cref{fig:skew-lvl-3}, 
we create a valid execution in \Cref{fig:skew-lvl-model}; 
the execution is created by applying transaction $\TTransaction_1$ followed by $\TTransaction_2$, 
thus $\ExecOf{\ComTrans{\TTransactionSetOf}}{\TExecution}$ holds. 
The execution is also valid with regard to the commit test for snapshot isolation, 
as this is one of the two executions for the \emph{write skew} example we presented in \Cref{fig:SI_model_example}.
Creating the other execution, by applying $\TTransaction_2$ followed by $\TTransaction_1$ instead, 
is also a valid choice, as implementation correctness merely asserts the existence of some valid execution.

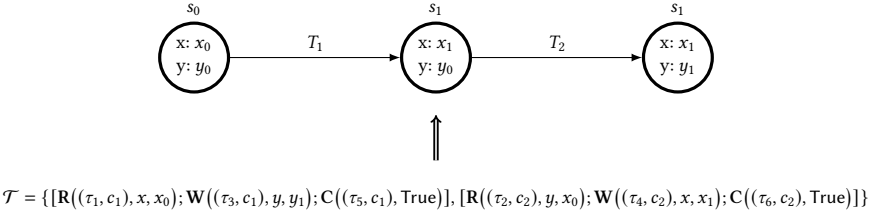
\begin{figure}[h]
  \vspace{-2mm}
  \small
  \centering
  \begin{tikzpicture}[>=latex, scale=0.8, every node/.style={scale=0.8}, node/.style={circle, draw=black, very thick, minimum size=11mm, align=left}]
    \node[node, label={$s_0$}] at (-4.0,0) (zero) {x: $x_0$ \\ y: $y_0$};
    \node[node, label={$s_1$}] at (0,0) (fst) {x: $x_1$ \\ y: $y_0$};
    \node[node, label={$s_1$}] at (4.0,0) (snd) {x: $x_1$ \\ y: $y_1$};
    \draw[->] (zero) -- (fst) node [above, pos=0.5] {$\TTransaction_1$};
    \draw[->] (fst) -- (snd) node [above, pos=0.5] {$\TTransaction_2$};
    \node[below=33pt of fst]
    (mat3) {$\TTransactionSetOf = \{[\TR\big((\tau_1, c_1), x, x_0\big);\TW\big((\tau_3, c_1), y, y_1\big);\TC\big((\tau_5, c_1), \TRUE\big)], 
      [\TR\big((\tau_2, c_2), y, x_0\big);\TW\big((\tau_4, c_2), x, x_1\big);\TC\big((\tau_6, c_2), \TRUE\big)]\}$};
    \draw[-{Implies},double, line width=0.25mm] (0,-1.7) -- (0,-0.95);
  \end{tikzpicture}
  \caption{Extraction of the \emph{write skew} example from level 3 to the model level.}
  \label{fig:skew-lvl-model}
  \vspace{1mm}
\end{figure}

Now that we have defined the extraction method, we can formally state our main theorem
that gives us implementation correctness for a database implementation verified with the specifications of \cite{SI-placeholder}
without making any additional assumption.
As our method for extracting transactions from code is stated as an invariant on traces, we get that any client interaction with a database 
verified using our method, will at any point in its execution only have produced valid transactions according to the state-based model. 
This is exactly what our main theorem, \Cref{def:adequacy} states.

\begin{theorem}[Adequacy]
  Let $\TLevel$ be an isolation level and let $\Spec_{\TLevel}$ be the specifications ascribed by \citet{SI-placeholder} for isolation level $\TLevel$.
  Let $e_1 \dots e_n$ be $n$ programs, \ie the main thread of $n$ nodes running clients of the database, such that for each $e_i$ we have that the quantified Hoare triple $\forall \mathit{lib},\; \anhoare{\Spec_{\TLevel}(\mathit{lib})}{e_i \oplus \mathit{lib}}{\TRUE}{}{}$ holds where $\oplus$ links the transactional library with its client.
  Finally, let $\mathit{lib}_{\TLevel}$ be some database implementation for which $\Spec_{\TLevel}(\mathit{lib}_{\TLevel})$ holds.
  Then, for any given initial state (of heaps and global network) $\sigma_0$ and any point in program execution we reach, \ie for any list of expressions $e'_{1,1}$ \dots{} $e'_{1,m_1}, e'_{2,1}$ \dots{} $e'_{2,m_2}$ \dots{} $e'_{n,1}$ \dots{} $e'_{n,m_n}$ (one for each thread in each node), any state $\sigma$, and any trace $\Trace$ such that
  \begin{center}
    \small
    $([e_1 \oplus \Wrap{\mathit{lib}_{\TLevel}}] \dots [e_n \oplus \Wrap{\mathit{lib}_{\TLevel}}], \sigma_0, \varepsilon) \longrightarrow^* ([e'_{1,1} \dots e'_{1,m_1}], [e'_{2,1} \dots e'_{2,m_2}]\dots[e'_{n,1} \dots e'_{n,m_n}], \sigma, \Trace)$    
  \end{center}
  we have $\ValTrace{\TLevel}{\Trace}$ holds, and furthermore, each $e'_{i,j}$ is not a crashed (runtime error) state.
  \label{def:adequacy}
\end{theorem}

The main theorem is a general result for any weak isolation level
$\TLevel$ with an implementation $\mathit{lib}_{\TLevel}$ which satisfies the
specifications $\Spec_{\TLevel}(\mathit{lib}_{\TLevel})$ ascribed by
\citet{SI-placeholder} that we presented for snapshot isolation in \Cref{sec:3}. 
The theorem states that any verified client of the database
running on $n$ separate nodes (each starting with only a single thread $e_i$)
when run for arbitrary number of steps ($\longrightarrow$ is the relation
$\nhstep_{(\cdot)}$ lifted to the level of the entire network) reaches a
configuration in which none of the threads is in a crash state (guaranteeing
no runtime error), and furthermore, the produced trace $t$ is a valid trace,
$\ValTrace{\TLevel}{\Trace}$, as per \eqref{fig:global_invariant}.
The latter ensures
implementation correctness as in \Cref{def:implements}. 
The crucial assumption in this theorem 
is that clients are verified against an arbitrary implementation that
satisfies the expected specifications. This generality allows us to show that a client
also can be linked with the wrapped library ($\Wrap{\mathit{lib}_{\TLevel}}$), 
which includes the code instrumentation of \Cref{fig:instrumentation}
to produce pre and post events, as long as the wrapped library 
also satisfies the expected specification.
Recall that the code instrumentation bears no effect
on programs' execution and can be erased before compilation. 
The crux of the proof
of \Cref{def:adequacy}, which we will discuss in \Cref{sec:5}, is to show two things,
one that
$\forall \mathit{lib}, \anhoare{\Spec_{\TLevel}(\mathit{lib})}{e_i \oplus \mathit{lib}}{\TRUE}{}{}$
implies
$\forall \mathit{lib}, \anhoare{\Spec_{\TLevel}(\Wrap{\mathit{lib}})}{e_i \oplus \Wrap{\mathit{lib}}}{\TRUE}{}{}$, and that $\Spec_{\TLevel}(\mathit{lib})$ implies $\Spec_{\TLevel}(\Wrap{\mathit{lib}})$, for any database implementation $\mathit{lib}$.
While the former is crucial for the proof of \Cref{def:adequacy}, it is trivial to prove.
However, the latter is the main challenge and forms the bulk of the proof of \Cref{def:adequacy}.
In particular, proving this requires maintaining \eqref{fig:global_invariant}, $\ValTrace{\TLevel}{\Trace}$ throughout the proof.
Intuitively, $\Spec_{\TLevel}(\mathit{lib})$ implies $\Spec_{\TLevel}(\Wrap{\mathit{lib}})$ because the specification $\Spec_{\TLevel}(\Wrap{\mathit{lib}})$ changes the  
internal definition of the resources used in $\Spec_{\TLevel}$, without making any external modifications, 
\ie ``the client interface of the specification'' is preserved. 
This can be done as all the resources in the per-operation specifications of 
$\Spec_{\TLevel}$ are existentially quantified as we emphasized in \Cref{sec:3.2}; 
all a client knows about the resources is that 
they work with the specifications. 
The wrapping of a specification thus means creating a new instantiation of the existentially quantified resources
used in the specifications, and showing that all the Hoare triples still hold with the new resources.
Having a proof of $\anhoare{{\Spec_{\TLevel}}(\Wrap{\mathit{lib}})}{e_i \oplus \Wrap{\mathit{lib}}}{\TRUE}{}{}$
gives us that all the clients can safely execute with code instrumentation 
while the invariant $\ValTrace{\TLevel}{\Trace}$ is upheld 
--- otherwise a client would get stuck according to the operational semantic and safety would be violated.
As $\ValTrace{\TLevel}{\Trace}$ encapsulates our extraction method presented throughout this section, 
implementation correctness follows and gives us the result of the main theorem.
The construction of the new resources used in $\Spec_{\TLevel}(\Wrap{\mathit{lib}})$ is made
as a part of the proof of \Cref{def:adequacy}; we show the proof structure 
of \Cref{def:adequacy} in \Cref{sec:5} and a proof sketch with more details for snapshot isolation in \Cref{sec:6}.


%% file: sections/structure.tex
\section{Proof Structure: Transactional Library Implements Transactional Model}
\label{sec:5}

In this section, we discuss the structure of the proof of \Cref{def:adequacy}. 
In \Cref{sec:4}, when we presented the extraction method leading to \Cref{def:adequacy} we used a particular execution of the \emph{write skew} example to explain how the extraction works.
However, to prove \Cref{def:adequacy} we must consider all executions of all possible clients.
This can formally be seen in our \Cref{def:lemma-wrap} which states that if $\Spec_{\TLevel}(\mathit{lib})$ holds then so does $\Spec_{\TLevel}(\Wrap{\mathit{lib}})$.
This intuitively captures that all executions of clients of $\Wrap{\mathit{lib}}$ satisfy \eqref{fig:global_invariant} as long as those clients are verified against $\Spec_{\TLevel}(\mathit{lib})$.

As we explained in \Cref{sec:4}, the crux of the proof of \Cref{def:adequacy} consists of the following \Cref{def:lemma-triple} and \Cref{def:lemma-wrap}:

\begin{lemma}[]
  \label[lemma]{def:lemma-triple}
  Let $\TLevel$ be an isolation level, and let $e_1 \dots e_n$ be n programs, then we have the following:
  \begin{center}
    $(\forall \mathit{lib},\; \anhoare{\Spec_{\TLevel}(\mathit{lib})}{e_i \oplus \mathit{lib}}{\TRUE}{}{}) \rightarrow
    \forall \mathit{lib},\;\anhoare{\Spec_{\TLevel}(\Wrap{\mathit{lib}})}{e_i \oplus \Wrap{\mathit{lib}}}{\TRUE}{}{}$.
  \end{center}
\end{lemma}
\begin{proof}
  This statement is a tautology of higher-order separation logic.
  Thus, the proof is rather straightforward and does not depend on specifics of $\Spec_{\TLevel}$, any $e_i$, or the wrap function.
  We simply need to introduce the antecedent and the universal quantification over $\mathit{lib}$, and instantiate the antecedent with $\Wrap{\mathit{lib}}$.
\end{proof}

\begin{theorem}[Main Theorem]
  \label{def:lemma-wrap}
  Let $\TLevel$ be an isolation level and $\mathcal{L}_\TLevel \eqdef \{\Trace ~ | ~ \ValTrace{\TLevel}{\Trace} \}$, then 
  \begin{align*}
    \forall lib. ~ \mathit{spec}_{\TLevel}(\mathit{lib}) \wand \Ttrace{\varepsilon} \wand \TtraceInv{\mathcal{L}_\TLevel} \wand \mathit{spec}_{\TLevel}(\Wrap{\mathit{lib}}).
  \end{align*}
\end{theorem}

With \Cref{def:lemma-triple} and \Cref{def:lemma-wrap} formally stated, we first present a sketch of the proof of \Cref{def:adequacy} before proceeding to discuss the proof of \Cref{def:lemma-wrap}.

\begin{proof}[A sketch of the proof of \Cref{def:adequacy}]
  In \Cref{def:adequacy}, we get to assume that the Hoare triples 
  $\forall \mathit{lib}, \anhoare{\Spec_{\TLevel}(\mathit{lib})}{e_i \oplus \mathit{lib}}{\TRUE}{}{}$ hold for all client programs $e_1 \dots e_n$.
  This, by \Cref{def:lemma-triple} allows us to conclude $\forall \mathit{lib}, \anhoare{\Spec_{\TLevel}(\Wrap{\mathit{lib}})}{e_i \oplus \Wrap{\mathit{lib}}}{\TRUE}{}{}$ for $e_1 \dots e_n$.
  Now, to finish the proof of \Cref{def:adequacy}, we need to show the preconditions of these Hoare triples hold.
  This is because the adequacy theorem of \citet{free-theorems}, roughly speaking, states that if we prove Hoare triples for our programs, and we prove their preconditions, then these programs are safe to run, \ie{}, they do not crash, and that all invariants hold throughout the execution.
  Furthermore, the adequacy of \citet{free-theorems} allows us to pick any $\mathcal{L}$, and only requires the Hoare triples to be proven under the assumption that we have $\Ttrace{\varepsilon}$ and $\TtraceInv{\mathcal{L}}$.
  Thus, what remains is to show the preconditions, \ie{}, $\Spec_{\TLevel}(\Wrap{\mathit{lib}})$, but this follows from \Cref{def:lemma-wrap} because \citet{SI-placeholder} have proven $\mathit{spec}_{\TLevel}(\mathit{lib})$ --- our Rocq mechanization builds on that of \citet{SI-placeholder}.
\end{proof}

The proof of \Cref{def:lemma-wrap} is what we discuss in the rest of this section and the next.
The proof of $\Spec_{\TLevel}(\Wrap{\mathit{lib}})$ in 
\Cref{def:lemma-wrap} follows the structure imposed of $\Spec_{\TLevel}$ that we presented in \Cref{sec:3.1}: 
we need to show the existence of a number of predicates and that all operations satisfy their specifications for these predicates.
As the operations of $\Wrap{\mathit{lib}}$ are the wrapped versions of the operations in $\mathit{lib}$ using the code instrumentation displayed in \Cref{fig:instrumentation}, 
a part of proving the Hoare triples for these operations is showing that the events emitted into the global trace $\Trace$, tracked using the resource $\Ttrace{t}$, satisfy the trace invariant $\TtraceInv{\mathcal{L}_\TLevel}$, \ie{}, $t \in \mathcal{L}_\TLevel$, at all times.
Remember that the trace invariant encapsulates the extraction of \Cref{sec:4} that 
we defined with the predicate $\ValTrace{\TLevel}{\Trace}$ as $\mathcal{L}_\TLevel$.

While the proof of \Cref{def:lemma-wrap} varies for each isolation level for a number of reasons, such as the difference in
per-operation specifications and the commit test, one of the main difficulties all isolation levels deal with is the insertion of linearization points.
Specifically, insertions into the sequence at level 2 in between the code instrumentation emitting the pre and post events at level 1. To this end, we set up a so-called \emph{ghost theory}.
A ghost theory consists of a number of propositions defined in terms of ownership of separation logic resources and a number \emph{laws} that govern them; see \citep{irisjournal} for details. 
The ghost theory we introduce for handling insertion of linearization points introduces predicates $\mathsf{Reflect_\TLevel}(t, \LTrace)$, $\PerLin(\tau)$, and $\PerPost(\tau)$.
The laws of our ghost theory are displayed in \Cref{fig:ghost-theory}.
The predicate $\mathsf{Reflect_\TLevel}(t, \LTrace)$ reflects the trace of pre and post events $t$ and sequence of linearization points $\LTrace$ to the level of the logic using separation logic resources.
The definition of $\mathsf{Reflect_\TLevel}(t, \LTrace)$ varies for each isolation level, but the laws in \Cref{fig:ghost-theory} are common laws that hold for all isolation levels.
The details of how $\mathsf{Reflect_\TLevel}(t, \LTrace)$ is defined are not important here. What matters is that it allows us to reason about traces, linearization points, sets of transactions and executions in our logic --- the sets of transactions and executions are existentially quantified inside $\mathsf{Reflect_\TLevel}(t, \LTrace)$. 
When proving the per-operation specification of $\Spec_{\TLevel}(\Wrap{\mathit{lib}})$ in \Cref{def:lemma-wrap},
we take the $\Ttrace{t}$ resource, tracking the global proof trace, and create an internal 
proof invariant
\begin{align*}
  \Inv{\exists t, \LTrace.\; \Ttrace{t} \ast \mathsf{Reflect_\TLevel}(t, \LTrace)}
\end{align*}
together with $\mathsf{Reflect_\TLevel}(t, \LTrace)$.
This invariant is available to us when we reason about emitting events
into the global trace 
that the instrumented operations do. When the global trace changes, we thus also have 
to update $\mathsf{Reflect_\TLevel}(t, \LTrace)$ using the laws of our ghost theory.
Now, the point of having $\mathsf{Reflect_\TLevel}(t, \LTrace)$ in an invariant together with $\Ttrace{t}$ 
is that $\mathsf{Reflect_\TLevel}(t, \LTrace)$ implies $ \ValTrace{\TLevel}{\Trace}$ (by \ruleref{Law 1})
which allow us to show that $t$ is in $\mathcal{L}_\TLevel$ at all times.
\begin{figure}[h]
  \vspace{-3mm}
  \begin{mathpar}
    \axiomH{Law 1}
    {\mathsf{Reflect_\TLevel}(t, \LTrace) \proves 
     \mathsf{Reflect_\TLevel}(t, \LTrace) \ast \ValTrace{\TLevel}{\Trace}}
    \and
    \inferH{Law 2}
    {\tau \not\in \TTags(t) \and \PreFunc(\langle\tau, \val\rangle)}
    {\mathsf{Reflect_\TLevel}(t, \LTrace) 
     \proves \mathsf{Reflect_\TLevel}(t \cdot \langle\tau, \val\rangle, \LTrace) \ast \PerLin(\tau)}
    \and
    \inferH{Law 3}
    {\PostFunc(\langle\tau, \val\rangle) \and 
     \CohTrace{t \cdot \langle\tau, \val\rangle}{\LTrace}}
    {\mathsf{Reflect_\TLevel}(t, \LTrace) \ast \PerPost(\tau) \proves 
     \mathsf{Reflect_\TLevel}(t \cdot \langle\tau, \val\rangle, \LTrace)}
    \and
    \axiomH{Law 4}
    {\mathsf{Reflect_\TLevel}(t, \LTrace) \ast \PerLin(\tau) \proves 
     \mathsf{Reflect_\TLevel}(t, \LTrace) \ast \PerLin(\tau) \ast \tau \not\in \TTags(\LTrace)}
    \and
    \axiomH{Law 5}
    {\mathsf{Reflect_\TLevel}(t, \LTrace) \ast \PerPost(\tau) \proves 
     \mathsf{Reflect_\TLevel}(t, \LTrace) \ast \PerPost(\tau) \ast \neg (\exists event. ~ \TTag(event) \in \TTags(t))}
    \and
    \axiomH{Law 6}
    {\mathsf{Reflect_\TLevel}(t, \LTrace) \proves \mathsf{Reflect_\TLevel}(t, \LTrace) \ast \PrefixReflect(t, \LTrace)}
    \and
    \axiomH{Law 7}
    {\mathsf{Reflect_\TLevel}(t, \LTrace) \ast \PrefixReflect(t, \LTrace) \proves 
     \mathsf{Reflect_\TLevel}(t, \LTrace) \ast \PrefixReflect(t', \LTrace') \ast t \leq t' \ast \LTrace \leq \LTrace'}
  \end{mathpar}
  \caption{Excerpt of ghost theory.}
  \label{fig:ghost-theory}
\end{figure}
When the instrumented database operations insert a pre event or a post event into the global trace, 
we need to update $\mathsf{Reflect_\TLevel}(t, \LTrace)$ accordingly.
These updates can be made using \ruleref{Law 2} and \ruleref{Law 3}, respectively.
As the code instrumentation, \ie the $\Temit$ and $\Tfresh$ expressions, are atomic operations in Iris, we can open the internal proof invariant 
around them --- an atomic operation is one that steps to a value in a single 
step of the operational semantics; see \Cref{fig:ghost_code_op}.
Thus, we can update $\Ttrace{t}$ with the 
code instrumentation specifications in \Cref{fig:ghost_code_ht}, 
and update $\mathsf{Reflect_\TLevel}(t, \LTrace)$ using the ghost theory laws.
When inserting the pre event with \ruleref{Law 2}, we gain a resource $\PerLin(\tau)$.
This is an exclusive permission to update the sequence of linearization points at level 2 by inserting a linearization point with tag $\tau$. 
Having $\PerLin(\tau)$, using \ruleref{Law 4}, we can ensure that no other linearization point has been inserted with tag $\tau$.
The insertion of linearization points into the sequence $\LTrace$ at level 2 of the extraction 
requires operation and isolation level specific reasoning to update $\mathsf{Reflect_\TLevel}(t, \LTrace)$.
The reason for this is that the insertion of the linearization point of operations enforces changes in
the set of transactions and execution that $\mathsf{Reflect_\TLevel}(t, \LTrace)$ 
tracks, and which it uses to establish $\ValTrace{\TLevel}{\Trace}$ in \ruleref{Law 1}.
In \Cref{sec:6}, we show how $\mathsf{Reflect_\TLevel}(t, \LTrace)$ can be updated 
with a linearization point for the $\mathsf{read}$ operation in snapshot isolation.
This involves the creation of new resources to instantiate the existentially quantified resources in $\Spec_{\TLevel}(\Wrap{\mathit{lib}})$.
Recall that it is due to the $\mathsf{read}$ operation being specified with an atomic Hoare triple
that we can open our internal proof invariant and insert a linearization point at level 2:
An atomic Hoare triple specification asserts that the operation is safe to execute and exposes the linearization point of the program at which point invariants can be accessed. 
When inserting linearization points, 
we also need to establish that the pre event for the linearization point we are inserting has 
been emitted to the trace $\Trace$.
This is done by creating a $\mathsf{Prefix}$ resource using \ruleref{Law 6} at the time when the pre event is inserted by the code instrumentation. 
Then, using \ruleref{Law 7}, when inserting the linearization point we can 
assert that the pre event is already in the trace of events at level 1. 
A similar reasoning applies when the code instrumentation inserts the post event; we need 
to show the existence of a corresponding linearization point in the global trace.

In the next section, we give more details of how a per-operation specification of $\Spec_{\TLevel}(\Wrap{\mathit{lib}})$
is proven for a particular operation, \ie{}, the $\mathsf{read}$ operation 
in snapshot isolation.


%% file: sections/proof.tex
\section{Proof Sketch: Transactional Library Implements Transactional Model}
\label{sec:6}

In this section, we present a proof of the read operation part of \Cref{def:lemma-wrap} for the snapshot isolation level using the ghost theory from \Cref{sec:5}.
The proof for other operations and other isolation levels are similar in nature; see our Rocq mechanization for more details.
We begin by discussing the important point that no
client can distinguish between $\mathit{lib}$ and $\Wrap{\mathit{lib}}$, because all the client relies on is that $\Spec_{\mathit{SI}}$ holds, and $\Spec_{\mathit{SI}}$ existentially quantifies all the resources that are used in the specification of operations.
To understand this subtle point, let us look at the specification of
the $\mathsf{read}$ operation in $\Spec_{\mathit{SI}}$ for snapshot isolation:
{
\begin{align}
  \anhoareatomic
    { \begin{array}[t]{@{}l@{}}
      \mapstoCache{\Key}{\Cst}{\ValueOption}
      \end{array}
    }
    { \begin{array}[t]{@{}l@{}}
      \SIrdC{\Cst}{\Key}
      \end{array}
    } 
    { \begin{array}[t]{@{}l@{}}
      \Ret \ValueOption. 
      \mapstoCache{\Key}{\Cst}{\ValueOption}
    \end{array}}
    {}{}
    \tag{SI-Read-Spec}
    \label{fig:atomic_read_spec}
\end{align}
}

\noindent
Intuitively, $\mapstoCache{\Key}{\Cst}{\ValueOption}$ is a logical proposition (resource)
stating that the key $\Key$ has value $\ValueOption$ in the (current) connection
$c$. That is, either $\ValueOption$ is the value of the last write to $\Key$ in
the current transaction of $\Cst$, or there is no prior write to $\Key$ in the
current transaction of connection $\Cst$, and $\ValueOption$ is the value of $\Key$ at the moment
(snapshot) when the transaction started. However, a proof of a client does
not in any way rely on the exact definition of the predicate
$\mapstoCache{\Key}{\Cst}{\ValueOption}$; to prove a client correct we \emph{only}
need to know that \eqref{fig:atomic_read_spec} holds, for some abstract definition
of $\mapstoCache{\Key}{\Cst}{\ValueOption}$. 
This is the parametric property of the specification structure we highlighted in \Cref{sec:3.2}, \ie{} 
the fact that the resources of the specification are existentially quantified.
Thus, we can change the
underlying definition of $\mapstoCache{\Key}{\Cst}{\ValueOption}$ in
$\Spec_\mathit{SI}$ without affecting the proof of the clients that are verified against it.
This is what our resource wrapping $\Wrapres{\mapstoCache{\Key}{\Cst}{\ValueOption}}$ does: it changes the underlying definition of
all the predicates involved in $\mapstoCache{\Key}{\Cst}{\ValueOption}$ without
changing its ``interface''. That is, in
$\Spec_\mathit{SI}(\Wrap{\mathit{lib}})$, the specification for the $\mathsf{read}$
operation is as follows:
{
\begin{align}
  \anhoareatomic
    { \begin{array}[t]{@{}l@{}}
      \Wrapres{\mapstoCache{\Key}{\Cst}{\ValueOption}}
      \end{array}
    }
    { \begin{array}[t]{@{}l@{}}
      \Wrap{\SIrdC{\Cst}{\Key}}
      \end{array}
    } 
    { \begin{array}[t]{@{}l@{}}
      \Ret \ValueOption. 
      \Wrapres{\mapstoCache{\Key}{\Cst}{\ValueOption}}
    \end{array}}
    {}{}
    \tag{SI-Read-Spec-Wrapped}
    \label{fig:atomic_read_spec_wrapped}
\end{align}
}

\noindent
where $\Wrap{\SIrdC{\Cst}{\Key}}$ is exactly as shown in \Cref{fig:extraction_method}.
Thus, to establish \Cref{def:lemma-wrap}, we need to show that \eqref{fig:atomic_read_spec} implies \eqref{fig:atomic_read_spec_wrapped}, and similarly for the other database operations: $\mathsf{init}$, $\mathsf{start}$, $\mathsf{write}$, and $\mathsf{commit}$.
In this section, we focus on the $\mathit{read}$ operation and refer the reader to the Rocq formalization for details about the other operations.
The proof obligation for the $\mathit{read}$ operation is formally stated as follows 
(note that the original read specification \eqref{fig:atomic_read_spec} is included in the assumption $\Spec_\mathit{SI}(lib_\mathit{SI})$):
\begin{align}
  & \Spec_\mathit{SI}(lib_\mathit{SI}) \ast \TtraceInv{\mathcal{L}_\mathit{SI}} \ast \tag{Read-Implies-Wrapped} \label{eq:read-implication}
  \Inv{\exists t, \LTrace.\; \Ttrace{t} \ast \mathsf{Reflect_\mathit{SI}}(t, \LTrace)} \\ 
  & \vdash \anhoareatomic
      {
        \begin{array}[t]{@{}l@{}}
          \Wrapres{\mapstoCache{\Key}{\Cst}{\ValueOption}} 
          \end{array}
      }
      {\Wrap{\mathsf{read} ~ c ~ k}}
      {
        \begin{array}[t]{@{}l@{}}
          \Ret \ValueOption. 
          \Wrapres{\mapstoCache{\Key}{\Cst}{\ValueOption}} 
        \end{array}
      }
      {}{}
      \notag
\end{align}
Notice here that the internal proof invariant 
$\Inv{\exists t, \LTrace.\; \Ttrace{t} \ast \mathsf{Reflect_\mathit{SI}}(t, \LTrace)}$
that we presented in \Cref{sec:5} is included in the assumptions, as is also the trace invariant $\TtraceInv{\mathcal{L}_\mathit{SI}}$
asserting that the global trace $t$ must remain in $\mathcal{L}_\mathit{SI}$, \ie $t \in \mathcal{L}_\mathit{SI}$
(recall $\mathcal{L}_\mathit{SI} \eqdef \{\Trace ~ | ~ \ValTrace{\mathit{SI}}{\Trace} \}$).
To make sure the global trace stays in $\mathcal{L}_\mathit{SI}$ during the proof of \eqref{eq:read-implication}, 
the crucial point is to get the resource wrapping $\Wrapres{\mapstoCache{\Key}{\Cst}{\ValueOption}}$ right.
It must be strong enough to show all the properties 
that we list below about the set of transactions $\TTransactionSetOf$ and execution $\TExecution$.
As we explained in \Cref{sec:5}, $\TTransactionSetOf$ and $\TExecution$ are existentially quantified in $\mathsf{Reflect_\mathit{SI}}(t, \LTrace)$.
The proposition $\mathsf{Reflect_\mathit{SI}}(t, \LTrace)$ ensures that the set of transactions $\TTransactionSetOf$ and execution $\TExecution$ that it tracks internally are coherent with $t$ and $\LTrace$ such that $\ValTrace{\mathit{SI}}{\Trace}$ can be satisfied --- this is why we can prove $\ValTrace{\mathit{SI}}{\Trace}$ in \ruleref{Law 1}.
The proposition $\Wrapres{\mapstoCache{\Key}{\Cst}{\ValueOption}}$ must be defined in such a way as to ensure that the following properties hold for the set of transactions $\TTransactionSetOf$ and execution $\TExecution$ that $\mathsf{Reflect_\mathit{SI}}(t, \LTrace)$ tracks internally:
\begin{enumerate}
  \item There must exist a $\mathit{start}$ linearization point with connection c in $\LTrace$ after which there is no $\mathit{commit}$ linearization point
  with connection c.
  \label{eq:listing-1}
  \item If the operation is not the first of the current transaction, then there is an open transaction in $\TTransactionSetOf$ holding the previous operations. 
  \label{eq:listing-2}
  \item If there is a previous $\TW(\mathit{id}, k, v)$ operation in the current transaction, then $\mathit{ov} = \Some v$, otherwise $\mathit{ov}$ is the value of key $k$ in a state 
  $s$ corresponding to the current transaction's snapshot in $\TExecution$. 
  \label{eq:listing-3}
\end{enumerate}

\noindent To create the wrapping $\Wrapres{\mapstoCache{\Key}{\Cst}{\ValueOption}}$,
we define a proposition $\mathsf{LatestWrite}(k, \mathit{ov}, c)$ using separation logic resources and extend our ghost theory.
This proposition works in conjunction with 
$\mathsf{Reflect_\mathit{SI}}(t, \LTrace)$ and lets us establish the properties listed above. 
With $\mathsf{LatestWrite}(k, \mathit{ov}, c)$, we can prove \ruleref{Law 8} in our ghost theory which we use when we insert the linearization point for the $\mathit{read}$ operation.
The wrapping of $\Wrapres{\mapstoCache{\Key}{\Cst}{\ValueOption}}$ is then defined to hold $\mathsf{LatestWrite}(k, ov, c)$ 
together with the original non-wrapped version of the resource $\mapstoCache{\Key}{\Cst}{\ValueOption}$ from $\Spec_\mathit{SI}(\mathit{lib}_\mathit{SI})$:
\begin{align}
  \Wrapres{\mapstoCache{\Key}{\Cst}{\ValueOption}} \eqdef{} \mapstoCache{\Key}{\Cst}{\ValueOption} \ast 
  \mathsf{LatestWrite}(k, \mathit{ov}, c)
  \tag{Wrapped-Key} \label{eq:wrapped-key}
\end{align}
\noindent Now, by \eqref{eq:wrapped-key} above, we can show the
implication \eqref{eq:read-implication} using our combined ghost theory of \Cref{fig:ghost-theory} and \ruleref{Law 8}. 
\begin{figure}[h]
    \vspace{-4mm}
    \begin{mathpar}
    \inferH{Law 8}
    {\mathsf{Reflect_\mathit{SI}}(t, \LTrace) \ast \PerLin(\tau) \ast \mathsf{LatestWrite}(k, \mathit{ov}, c)}
    {\mathsf{Reflect_\mathit{SI}}\big(t, \LTrace \cdot \mathsf{\TR}((\tau, c), k, \mathit{ov})\big) \ast 
     \PerPost(\tau) \ast \mathsf{LatestWrite}(k, \mathit{ov}, c)}
    \end{mathpar}
    \vspace{-7mm}
\end{figure}
As highlighted in \Cref{sec:5}, we need to take care of the fact that the wrapped code contains 
code instrumentation
which inserts pre- and post-events into the global trace. 
Thus, to ensure that the global
trace remains within $\mathcal{L}_\mathit{SI}$, we need to insert a linearization point into
the trace of linearization points (see condition
\eqref{item:pre-pos-must-have-lin} of $\CohTrace{\Trace}{\LTrace}$ on page
\pageref{item:pre-pos-must-have-lin}). Crucially, the Hoare triple
\eqref{fig:atomic_read_spec} is a logically atomic Hoare triple.
The upshot of this is that (the specification of) the
$\mathsf{read}$ operation can be treated atomically, \ie we can access invariants
when reasoning about the $\mathsf{read}$ operation --- recall that this access to invariants in practice takes place around the linearization point of the expression for which we have the atomic Hoare triple. Thus, before (the linearization point
of) the call to $\SIrdC{\Cst}{\Key}$ inside $\Wrap{\SIrdC{\Cst}{\Key}}$, we can
assume that
$\exists t, \LTrace.\; \Ttrace{t} \ast
\mathsf{Reflect_\mathit{SI}}(t, \LTrace)$ holds, and we
need to reestablish it afterwards. This allows us to insert a linearization
point for the $\mathsf{read}$ operation into $\LTrace$, which in turn means that
when the code instrumentation inserts the post event, the trace remains coherent, \ie{},
$\CohTrace{\Trace}{\LTrace}$ continues to hold. 
As we get to assume the wrapped resource $\Wrapres{\mapstoCache{\Key}{\Cst}{\ValueOption}}$
around the insertion of the linearization point, we can split $\Wrapres{\mapstoCache{\Key}{\Cst}{\ValueOption}}$
into the non-wrapped version $\mapstoCache{\Key}{\Cst}{\ValueOption}$, to satisfy the precondition of $\SIrdC{\Cst}{\Key}$, 
and the resource we added $\mathsf{LatestWrite}(k, \mathit{ov}, c)$, 
to update $\mathsf{Reflect_\mathit{SI}}(t, \LTrace)$ using \ruleref{Law 8} of our ghost theory.

To summarize, the proof of \Cref{def:lemma-wrap}, which we use to establish \Cref{def:adequacy}, amounts to doing what we have done
in this section for the $\mathsf{read}$ operation for all the operations of the database interface. 
The complexity here lies in coming up with wrapped resource definitions that let 
us raise per-operation specifications to the level of abstraction of a global execution.
The crux of the difficulty, 
stems from all the layers of our extraction method in \Cref{fig:extraction_method} being tightly connected  
through the coherence predicates of the global proof invariant. 
Updating one layer, using the reflection in the logic ($\mathsf{Reflect_\mathit{SI}}(t, \LTrace)$), requires 
reestablishing assertions about all the other layers internal to $\mathsf{Reflect_\mathit{SI}}(t, \LTrace)$.
In fact, the proof of \Cref{def:lemma-wrap} consists of several thousand lines of Rocq proof code
for \emph{each} isolation level.


%% file: sections/related_work.tex
\section{Related Work}
\label{sec:7}

Recent work on the higher-order separation logic
Trillium \citep{Trillium}, which is based on Iris,
developed a method for showing that a labeled state transition systems is 
refined by an executable implementation. 
While Trillium is language generic, implementations can be written 
in AnerisLang \citep{DBLP:conf/esop/Krogh-Jespersen20}. 
In particular, they prove that AnerisLang implementations of
two-phase commit and single-degree Paxos refine their respective \TLA models. 
The refinement is then used to carry over properties to the implementations 
that were established at the model level in \TLA.
A natural question to ask is why we did not use Trillium to reason about the 
state-based model for database implementations, as the state-based model 
has been embedded in \TLA \citep{Soethout2021}.
The reason is that with Trillium, we can gain that a database implementation refines the 
model, but we do not get a main theorem, such as \Cref{def:adequacy}, stating that any 
client interaction will produce a valid set of transactions.
The \TLA model of \citet{Soethout2021} has to be model-checked against each
workload one wants to verify which quickly leads to state space explosion.

Much work is done on verifying distributed systems using separation logic 
without making connections to other models. 
Also in the context of AnerisLang,
\citet{abel/crdts, nieto_et_al:LIPIcs.ECOOP.2023.22} builds a framework 
for conflict-free replicated data types allowing users to define 
highly-available distributed applications, and 
\citet{DBLP:journals/pacmpl/GondelmanGNTB21} creates a distributed non-transactional 
key-value store guaranteeing causal consistency. 
Using Grove \citep{grove}, a distributed separation logic with support for reasoning
about node crashes built on the Iris framework, a distributed key-value store with primary backup replication, 
based upon a replicated state machine library, is verified \citep{grove}.
As the isolation level specifications of \citet{SI-placeholder} have been formalized in Aneris 
(Aneris is the program logic for reasoning about AnerisLang), 
the Aneris instantiation of Iris is the choice in this paper, 
but all of our techniques will also work in other instantiations of Iris such as Grove.

There is a line of work on creating databases with so-called \emph{verifiable transactions} 
\citep{zhao2023veritxn, xia2022litmus}. That is, the database produces at runtime a proof that it does  
not break guarantees such as isolation levels for transactions. 
This method naturally comes with runtime overhead, and it is also not guaranteed that the verification mechanisms 
inside the databases are implemented correctly. 


%% file: sections/conclusion.tex
\section{Conclusion and Future Work}
\label{sec:8}
This paper presents the first method for formally verifying that a database correctly implements  
its isolation level,
by lifting Hoare triples specifications to transactional consistency models using 
invariants on traces.
We have lifted specifications of the weak isolation levels 
read uncommitted, read committed and snapshot isolation, thus these results 
serve as free theorems for all future verification efforts; 
by verifying that a database implements the separation logic specifications, 
it automatically implements its isolation level. 

In the future, it will be interesting to see our approach
applied to more isolation levels than we do in this paper. 
There exists a plethora of other levels, such as \emph{opacity} \citep{guerraoui2008correctness} 
and \emph{virtual world consistency} \citep{imbs2012virtual},
which are popularized in software transactional memory systems and 
restricts the behavior of ongoing transactions,
or variants of snapshot isolation with intricate behavior such as 
\emph{parallel} snapshot isolation
\citep{cerone2015framework, sovran2011transactional}.
Currently, these isolation levels do not have separation logic specifications, 
which, of course, is a prerequisite for using our approach.
We do not expect the extension to other isolation levels to impose a major challenge, 
but it could involve new challenges to extend the approach to other database 
APIs; currently we model a key-value store, but SQL compliant systems 
support a much more complex API. The predominant transactional models we have used in this paper, 
do not fully support SQL: \citet{Crooks17} is limited to a key-value store while 
\citet{Adya99} does support \emph{predicate}-based operations that can be used to model features of SQL.


%% file: sections/example.tex
\section{Proving Safety of Examples}
\label{sec:example}

In this section, we show how to use the per-operation specifications to prove safety of the \emph{write skew} example. 
In \Cref{fig:write-skew-example-program}, we display the code of the example together with the invariant we use. 
The vertical bars separate two clients program each running their transaction of the \emph{write skew} example.
Initially, when using the snapshot isolation specification, we obtain initial points-to resources for the keys $x$ and $y$: 
$\mapstoMem{x}{\HistEmpty}$ and $\mapstoMem{y}{\HistEmpty}$. 
With these resources, we create the invariant $\knowInv{}{\exists h_x h_y,~ \mapstoMem{x}{h_x} \ast \mapstoMem{y}{h_y}}$.
The proof of each client transaction can be done in parallel --- this is a feature of concurrent/distributed separation logics.
As the invariant is a persistent resource, \ie it does not assert exclusive ownership, we can use it when proving 
both clients.
We go through the proof of the first transaction; the proof of the other transaction is similar.
Having initialized a connection using the operation $\mathsf{init}$ (this is not shown as part of \Cref{fig:write-skew-example-program}), 
the client holds the $\CanStart{\Cst}$ resource as specified in \ruleref{si-init-client-spec}. 
For the precondition of the $\mathsf{start}$ operation in \ruleref{si-start-spec}, we need the global points to resources. 
Because the $\mathsf{start}$ operation is specified using an atomic Hoare triple, we can make use of the rule in \ruleref{Inv-atomic}, 
and get access to the global points-to resources $\mapstoMem{x}{h_x} \ast \mapstoMem{y}{h_y}$ from inside the invariant.
Using the global points-to resources and the $\CanStart{\Cst}$ resource,
we can satisfy the precondition of the start specification. 
In the postcondition, we get back the unchanged global points to resources $\mapstoMem{x}{h_x} \ast \mapstoMem{y}{h_y}$, 
which we use to reestablish the invariant, and the resources 
$\Active{\Cst}{\{x \mapsto h_x; y \mapsto h_y\}} \ast \mapstoCache{x}{\Cst}{\HistVal{h_x}} \ast \mapstoCache{y}{\Cst}{\HistVal{h_y}} \ast \KeyUpdStatus{\Cst}{x}{\FALSE} \ast \KeyUpdStatus{\Cst}{y}{\FALSE}$, 
which we hold onto locally for the rest of the proof.
The body of the transaction, the $\mathsf{read}$ operation $"\SIrdC{y}"$ followed by the $\mathsf{write}$ operation $"\SIwrC{x}{1}"$, changes 
our local resources to 
$\Active{\Cst}{\{x \mapsto h_x; y \mapsto h_y\}} \ast \mapstoCache{x}{\Cst}{\Some 1} \ast \mapstoCache{y}{\Cst}{\HistVal{h_y}} \ast \KeyUpdStatus{\Cst}{x}{\TRUE} \ast \KeyUpdStatus{\Cst}{y}{\FALSE}$
through the \ruleref{si-read-spec} and \ruleref{si-write-spec} specifications;
the $\mathsf{read}$ operation make no changes to the resources, while the $\mathsf{write}$ operation is responsible 
for the update reflecting that the value 1 is written to the key $x$.
Notice how we do not need to access the resources of the invariant to reason about the body of the transaction. 
This is unlike the last operation, the $\mathsf{commit}$ operation, at which we again open the invariant 
to access the current state of the database: $\mapstoMem{x}{h_x}' \ast \mapstoMem{y}{h_y}'$, for some histories $h_x'$ and $h_y'$.
Together with our local resources, this is enough to satisfy the precondition of \ruleref{si-commit-spec}.
Whether the global points-to resources are updated or not, as part of the postcondition of \ruleref{si-commit-spec}, 
do not matter because we can always satisfy our invariant $\knowInv{}{\exists h_x h_y,~ \mapstoMem{x}{h_x} \ast \mapstoMem{y}{h_y}}$ 
since the histories are existentially quantified. This concludes the proof of the \emph{write skew} example. 
\begin{figure}[h]
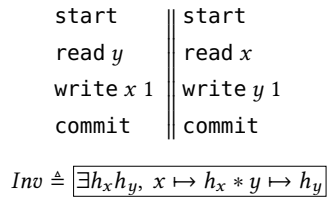

  \vspace{-3mm}
    \centering
    \begin{minipage}[b]{0.5\textwidth} 
    \begin{align*}
    \left.
    \begin{aligned}
      & \SIstart\\
      & \SIrdC{y}\\
      & \SIwrC{x}{1}\\
      & \SIcommit
    \end{aligned}
    ~\middle\Vert~
    {
    \begin{aligned}
      & \SIstart\\
      & \SIrdC{x}\\
      & \SIwrC{y}{1}\\
      & \SIcommit
    \end{aligned}
    }
    \right.
    \end{align*}
    \begin{align*}
      Inv &\eqdef{} \knowInv{}{\exists h_x h_y,~ \mapstoMem{x}{h_x} \ast \mapstoMem{y}{h_y}} 
    \end{align*}
    \caption{\emph{write skew} example program.}
    \label{fig:write-skew-example-program}
  \end{minipage}
\end{figure}

Proving safety of the \emph{write skew} example, does little to strengthen our assurance 
that the separation logic specifications correctly captures snapshot isolation;
snapshot isolation is known for allowing both transactions in the \emph{write skew} example to commit, 
as it is impossible for the two transactions to have a write conflict.
Thus, if we assert that both transaction always commit, 
we capture that no write conflicts can happen and therefore 
every execution is valid under snapshot isolation. 
Concretely, this amounts to proving a variant of the example in \Cref{fig:write-skew-example-program} 
with both $\mathsf{commit}$ operations are wrapped in an assert statement: $\SIassert{\SIcommit}$. 
This proof is more involved than the one we did above
as we must create an invariant strong enough to show 
that both transactions commit.  
Specifically, when both transactions try to commit, we need enough information to prove that 
when the $\SIcommit$ operation returns $\FALSE$, the proposition $\lnot \CanCommitPredicate{\map}{\Snapshot}{\Cache}$
does not hold.
Recall that the $\CanCommitPredicateSymb$ predicate is used to capture when there are no write conflicts, 
thus the negation $\lnot \CanCommitPredicate{\map}{\Snapshot}{\Cache}$ states that there is a write conflict.
To have enough information to derive a contradiction from $\lnot \CanCommitPredicate{\map}{\Snapshot}{\Cache}$, 
we use the following invariant:
\begin{align*}
  Inv &\eqdef{} \knowInv{}{\exists h_x h_y,~ 
  \mapstoMem{x}{h_x} \ast \mapstoMem{y}{h_y} \ast 
  \big(h_x = \HistEmpty \lor (h_x = \HistNonEmpty{1} \ast \FinishToken_1)\big) \ast
  \big(h_y = \HistEmpty \lor (h_y = \HistNonEmpty{1} \ast \FinishToken_2)\big)}.
\end{align*}
This invariant uses a disjunction for each of the histories $h_x$ and $h_y$ to track the values they contain. 
Focusing on $h_x$ ($h_y$ has the same structure), it is either the case that it represents  
the initial state $h_x = \HistEmpty$, or the first transaction has committed $(h_x = \HistNonEmpty{1} \ast \FinishToken_1)$.
Here, $\FinishToken_1$ is an exclusive token in Iris in the sense that having it twice is impossible, 
\ie $\FinishToken_1 \ast \FinishToken_1 \proves \FALSE$.
In the proof of the \emph{write skew} example with assertions around the $\mathsf{commit}$ operation, we initially create two tokens, 
$\FinishToken_1$ and $\FinishToken_2$, and take one for the proof of each of the clients.
Focusing on the client program with the first transaction, $\FinishToken_1$ can be used to disregard the disjunct 
$(h_x = \HistNonEmpty{1} \ast \FinishToken_1)$ each time the invariant is opened.
This means that the initial state is always observed until the first transaction commits and
updates the history of $x$ from $[]$ to $\HistNonEmpty{1}$, and puts the token in the invariant to satisfy 
the disjunct $(h_x = \HistNonEmpty{1} \ast \FinishToken_1)$.
Because the first transaction observes the global points-to resource of $x$ pointing to the empty history, $[]$, 
until it commits its own $\mathsf{write}$, 
we can assert that there is no write conflict with the $\CanCommitPredicateSymb$ predicate, 
as no other transaction has written to $x$. Therefore, the transaction must commit.


Prior to having mature transactional models in the literature for reasoning about transactional consistency, isolation levels were 
specified using phenomena: small examples that must be prohibited by particular levels \cite{orig-SI}. 
The \emph{write skew} example is one such phenomenon.
It must be prohibited by serializability, while snapshot isolation allows it, as we have seen.
\citet{SI-placeholder} justify that their specifications for isolation level $\TLevel$ capture level $\TLevel$'s guarantees by showing that their specifications rule out the phenomena that isolation level $\TLevel$ should.
The work of this paper lifts the assurance from phenomena to transactional models for actual program executions.
This is analogous to how it was done previously in the literature \cite{Adya99, Adya00} but without any links to program executions.
